\documentclass[11pt]{article}

\usepackage{amsmath,epsfig}
\usepackage{amssymb}
\usepackage{subfigure}
\usepackage{rotating}
\usepackage{tabularx}
\usepackage{multicol}
\usepackage{paralist}
\usepackage{multirow}
\usepackage{color}
\usepackage{verbatim}
\usepackage[ruled,vlined,shortend]{algorithm2e}
\usepackage{comment}
\usepackage{url}

\renewcommand{\emptyset}{\varnothing}

\newcommand{\atom}[1]{\underline{#1}}
\newcommand{\tuple}[1]{\mathbf{#1}}

\newcommand{\dep}{\Sigma}
\newcommand{\tdep}{\Sigma_T}

\newcommand{\kdep}{\Sigma_K}
\newcommand{\ndep}{\Sigma_\bot}

\newcommand{\isa}[1]{\mathit{ISA}}

\newcommand{\head}[1]{\mathit{head}(#1)}
\newcommand{\body}[1]{\mathit{body}(#1)}

\newcommand{\ans}[3]{\mathit{ans}(#1,#2,#3)}

\newcommand{\ins}[1]{\mathbf{#1}}

\newcommand{\insX}{\ins{X}}
\newcommand{\insY}{\ins{Y}}
\newcommand{\insZ}{\ins{Z}}

\newcommand{\chase}[2]{\mathit{chase}(#1,#2)}

\newcommand{\apchase}[3]{\mathit{chase}^{[#1]}(#2,#3)}

\newcommand{\mods}[2]{\mathit{mods}(#1,#2)}

\renewcommand{\paragraph}[1]{\textbf{#1}}

\newcommand{\A}{\mathcal{A}} 
 
\newcommand{\E}{\mathcal{E}} \newcommand{\F}{\mathcal{F}}

 \renewcommand{\L}{\mathcal{L}}

 \newcommand{\R}{\mathcal{R}}

\newcommand{\ra}{\rightarrow}

\newcommand{\la}{\leftarrow}

\newcommand{\tup}[1]{\langle #1\rangle}            

\def\qed{\hfill{\qedboxempty}      
  \ifdim\lastskip<\medskipamount \removelastskip\penalty55\medskip\fi}

\def\qedboxempty{\vbox{\hrule\hbox{\vrule\kern3pt
                 \vbox{\kern3pt\kern3pt}\kern3pt\vrule}\hrule}}

\def\qedfull{\hfill{\qedboxfull}   
  \ifdim\lastskip<\medskipamount \removelastskip\penalty55\medskip\fi}

\def\qedboxfull{\vrule height 4pt width 4pt depth 0pt}

\newcommand{\plusminus}{Datalog$^{\pm}$}

\newtheorem{theorem}{Theorem}

\newtheorem{lemma}[theorem]{Lemma}
\newtheorem{claim}[theorem]{Claim}

\newtheorem{definition}{Definition}
\newtheorem{example}{Example}

\begin{document}
\sloppy

\title{\textbf{Ontological Queries: Rewriting and Optimization (Extended
Version)}\footnote{This is an extended and revised version of the
paper~\cite{GoOP11}.}}

\date{}

\author{Georg Gottlob{\small $^{1,2}$},
Giorgio Orsi{\small $^{1,3}$}, Andreas Pieris{\small $^{1}$}\\
\vspace{0.4mm} \\
\fontsize{10}{10}\selectfont\itshape $~^{1}$Department of Computer Science, University of Oxford, UK\\
\fontsize{10}{10}\selectfont\rmfamily\itshape $~^{2}$Oxford-Man Institute of Quantitative Finance, University of Oxford, UK\\
\fontsize{10}{10}\selectfont\rmfamily\itshape $~^{3}$Institute for the Future of Computing, University of Oxford, UK\\
\vspace{0.4mm}\\
\fontsize{9}{9}\selectfont\ttfamily\upshape
\{georg.gottlob,giorgio.orsi,andreas.pieris\}@cs.ox.ac.uk}

\maketitle

\begin{abstract}
Ontological queries are evaluated against an ontology
rather than directly on a database. The evaluation and optimization
of such queries is an intriguing new problem for database research.
In this paper we discuss two important aspects of this problem:
query rewriting and query optimization. Query rewriting consists of
the compilation of an ontological query into an equivalent query
against the underlying relational database. The focus here is on
soundness and completeness. We review previous results and present a
new rewriting algorithm for rather general types of ontological
constraints. 
In particular, we show how a conjunctive query against an ontology can be compiled
into a union of conjunctive queries against the underlying database.
Ontological query optimization, in this context, attempts to improve
this process so to produce possibly small and cost-effective UCQ rewritings for an input query. We review existing optimization methods, and
propose an effective new method that works for {\em linear
{\plusminus}}, a class of Datalog-based rules that encompasses
well-known description logics of the \emph{DL-Lite} family.
\end{abstract}

\section{Introduction}\label{sec:introduction}

This paper is about ontological query processing, an important new
challenge to database research. We will review existing methods and
propose new algorithms for compiling  an {\em ontological query},
that is, a query against an ontology on top of a
relational database, into a direct query against this database, and
we will deal with optimization issues related to this process so as
to obtain possibly small and efficient compiled queries. In this
section, we first discuss a number of relevant concepts, and then
illustrate query rewriting and optimization processes in the context
of a small but non-trivial example.

\vspace{-2mm}

\paragraph{Ontologies.} The use of ontologies and ontological reasoning in companies,
governmental organizations, and other enterprises has become
widespread in recent years. An {\em ontology} is an {\em explicit
specification of a conceptualization} of an area of
interest~\cite{Gruber93}, and consists of {\em a formal
representation of knowledge as a set of concepts within a domain,
and the relationships between those concepts}~\cite{wiki-ontology}.
To distinguish an enterprise ontology from a data dictionary, Dave
McComb explicitly refers to the formal semantics of ontologies that
enables automated processing and inferencing, while the
interpretation of a data dictionary is strictly done by
humans~\cite{McComb06}. Moreover, ontologies have been adopted as
high-level conceptual descriptions of the data contained in data
repositories that are sometimes distributed and heterogeneous in the
data models. Due to their high expressive power, ontologies are also
substituting more traditional conceptual models such as UML
class-diagrams and E/R schemata.

\vspace{-2mm}

\paragraph{Description Logics.} Description logics (DLs) are logical
languages for expressing and modelling ontologies. The best known
DLs are those underlying the \emph{OWL}
language\footnote{http://www.w3.org/TR/owl2-overview/}. The main
ontological reasoning and query answering tasks in the complete OWL
language, called \emph{OWL Full}, are undecidable. For the most
well-known decidable fragments of OWL, ontological reasoning and
query answering is still computationally very hard, typically
2\textsc{exptime}-complete.

\vspace{-2mm}

In description logics, the ontological axioms are usually divided
into two sets: The {\em ABox} (assertional box), which essentially
contains factual knowledge such as ``IBM is a company'', denoted by
$\mathit{company}(\mathit{ibm})$, or ``IBM is listed on the
NASDAQ'', which could be represented as a fact of the form
$\mathit{list\_comp}(\mathit{ibm},\mathit{nasdaq})$, and a {\em
TBox} (terminological box) which contains axioms and constraints
that allow us, on the one hand, to infer new facts from those given
in the ABox, and, on the other hand, to express restrictions such
as keys. For example, a TBox may contain an axiom stating that for
each fact $\mathit{list\_comp}(X,Y)$, $Y$ must be a financial index,
which in DL is expressed as $\exists \mathit{list\_comp}^{-}
\sqsubseteq \mathit{fin\_idx}$. If the fact
$\mathit{fin\_idx}(\mathit{nasdaq})$ is not already present in the
ABox, it can be derived via the above axiom from
$\mathit{list\_comp}(\mathit{ibm},\mathit{nasdaq})$. Thus, the
atomic query ``$q(X) \leftarrow \mathit{fin\_idx}(X)$'' would return
$\mathit{nasdaq}$ as one of the answers. Note that the axiom
$\exists \mathit{list\_comp}^{-} \sqsubseteq \mathit{fin\_idx}$,
which corresponds to an inclusion dependency, is {\em enforced by
adding new tuples}, rather than just being {\em checked}. This is
one main difference between ontological constraints and classical
database dependencies. In database terms, the above axiom is to be
interpreted more like a {\em trigger} than a classical constraint.

\vspace{-2mm}

\paragraph{Ontology Based Data Access (OBDA).} We are currently witnessing
the marriage of ontological reasoning and database technology. In
fact, this amalgamation consists in the realization of the obvious
idea that ABoxes shall be implemented in form of a relational
database, or even stored in classical RDBMSs. Moreover, very large
existing databases are semantically enriched with ontological
constraints. There are a number of recent commercial systems and
experimental prototypes that extend RDBMSs with the
possibility of querying an ontology that is rooted in a database
(for examples, see Section~\ref{sec:related}). The main problem here
is how to couple these two  different types of technology smoothly
and efficiently, and this is also the main theme of the present
paper.

\vspace{-2mm}

One severe obstacle to efficient OBDA is the already mentioned high
computational complexity of query answering with description logics.
The situation clearly worsens when the ABoxes of enterprise
ontologies are very large databases. To tackle this problem, new,
lightweight DLs have been designed, that guarantee {\em
polynomial-time data complexity} for conjunctive query answering.
This means that based on a fixed TBox, a fixed query can be
answered in polynomial time over variable databases. The best-known
and best-studied examples of such lightweight DLs are the
\emph{DL-Lite}~\cite{CDLLLR07} and $\E\L$ (see,
e.g.,~\cite{Baader03}) families. These languages can be considered
tractable subclasses of OWL. It was convincingly argued that simple
DLs such as DL-Lite or $\E\L$ are sufficient for modelling an
overwhelming number of applications.

\vspace{-2mm}

More recently, the {\plusminus} family of description logics was
introduced~\cite{CGL09,CaGL09,CGP10,CaGP10}. Its syntax is based on
classical first-order logic, more specifically, on variants of the
well-known Datalog language~\cite{CGT89}. The basic Datalog$^\pm$
rules are known as \emph{tuple-generating dependencies} (TGDs) in
the database literature~\cite{BeVa84}. Tractable DLs in this
framework are {\em guarded {\plusminus}}, which is noticeably more
general than both DL-Lite and $\E\L$, and the DLs {\em linear
{\plusminus}} and {\em sticky-join {\plusminus}}, which both
encompass DL-Lite.

\vspace{-2mm}

Besides being more expressive than DL-Lite, suitable {\plusminus}
languages offer a more compact representation of the attributes of
concepts and roles, since description logics are usually restricted
to unary and binary predicates only. Consider, as an example, a
relation $\mathit{stock}(\underline{{\sf id}},{\sf name},{\sf
unit}$-${\sf price})$. Representing this relation in DL would
require the introduction of a concept symbol $\mathit{stock}$, and
of three attribute symbols $\mathit{id}$, $\mathit{name}$ and
$\mathit{unit}$-$\mathit{price}$. These entities must be then weaved
together by the TBox formula $\mathit{stock} \sqsubseteq \exists id
\sqcap \exists \mathit{name} \sqcap \exists
\mathit{unit}$-$\mathit{price}$. {\plusminus} represents the
relation in a natural way by means of a ternary predicate
$\mathit{stock}$. In the same way, {\plusminus} provides a more
natural syntax for many other DL formulae; for example, an inverse
role assertion $r \sqsubseteq s^{-}$ is represented as a (full) TGD
$r(X,Y) \rightarrow s(Y,X)$, while an existential restriction $p
\sqsubseteq \exists r.q$ is represented as a (partial) TGD $p(X)
\rightarrow \exists Y \, r(X,Y),q(Y)$.

\vspace{-2mm}

\paragraph{First-Order Rewritability.} Polynomial-time tractability is
often considered not to be good enough for efficient query
processing. Ideally, one would like to achieve the same complexity
as for processing SQL queries, or, equivalently, first-order (FO)
queries. An ontology language $\mathcal{L}$ is \emph{first-order rewritable} if, for every TBox $\dep$
expressed in $\mathcal{L}$ and a query $q$, a first-order query $q_{\dep}$ (called the {\em perfect rewriting}) can be constructed such that, given a database $D$, 
$q_{\dep}$ evaluated over $D$ yields exactly the same result as $q$ evaluated against
$D$ and $\dep$. Since answering first-order
queries is in the class \textsc{ac}$_{0}$ in data
complexity~\cite{Var95}, it immediately follows that under
FO-rewritable TGDs, query answering is also in \textsc{ac}$_{0}$ in
data complexity

\vspace{-2mm}

This notion was first introduced by Calvanese et al.~\cite{CDLLLR07}
in the concept of description logics. If a DL guarantees the
FO-rewritability of each query under every TBox, we simply say that
the logic is FO-rewritable. FO-rewritability is a most desirable
property since it ensures that the reasoning process can be largely
decoupled from data access. In fact, to answer query $q$, a separate
software can compile $q$ into $q_{\dep}$, and then just submit
$q_{\dep}$ as a standard SQL query to the DBMS holding $D$, where it
is evaluated and optimized in the usual way.

\vspace{-2mm}

Excitingly, it was shown that the members of the DL-Lite family, as
well as the slightly more expressive language linear {\plusminus}
are FO-rewritable. Moreover, even the much more expressive language
of sticky-join {\plusminus} is FO-rewritable. For these languages, a
pair $\tup{\Sigma,q}$, where $q$ is a CQ, is rewritten as an SQL
expression equivalent to a UCQ $q_{\dep}$. The research challenge we
address in this paper is precisely the question of how to rewrite
$\tup{\Sigma,q}$ to $q_{\dep}$ correctly and efficiently. Let us
illustrate this process by a small, but comprehensive example.

\vspace{-2mm}

Consider the following relational schema $\mathcal{R}$ representing
financial information about companies and their stocks:
\[
\begin{array}{rcl}
&& \mathit{stock}({\sf id},{\sf name},{\sf unit}$-${\sf price})\\
&& \mathit{company}({\sf name},{\sf country},{\sf segment})\\
&& \mathit{list\_comp}({\sf stock},{\sf list})\\
&& \mathit{fin\_idx}({\sf name},{\sf type},{\sf ref}$-${\sf mkt})\\
&& \mathit{stock\_portf}({\sf company},{\sf stock},{\sf qty}).
\end{array}
\]
The $\mathit{stock}$ relation contains information about stocks such
as the name, and the price per unit. The relation $\mathit{company}$
contains information about companies; in particular, the name, the
country, and the market segment of a company.
The relation $\mathit{list\_comp}$ relates a stock to a financial
index (i.e., NASDAQ, FTSE, NIKKEI) represented by the relation
$fin\_idx$ which, in turn, contains information about the types of
stocks in the index, and the reference market (e.g., London Stock
Exchange). Finally, $\mathit{stock\_portf}$ relates companies to
their stocks along with an indication of the amount of the
investment.

\vspace{-2mm}

Datalog$^\pm$ provides the necessary expressive power to extend
$\mathcal{R}$ with ontological constraints in an easy and intuitive
way. Examples of such constraints follow:

\vspace{-2mm}

{\small
\[
\begin{array}{rcl}
\sigma_1 &:& \mathit{stock\_portf(X,Y,Z)} \rightarrow \exists V \exists W\ \mathit{company(X,V,W)}\\
\sigma_2 &:& \mathit{stock\_portf(X,Y,Z)} \rightarrow \exists V \exists W\ \mathit{stock(Y,V,W)}\\
\sigma_3 &:& \mathit{list\_comp(X,Y)} \rightarrow \exists Z \exists W\ \mathit{fin\_idx(Y,Z,W)}\\
\sigma_4 &:& \mathit{list\_comp(X,Y)} \rightarrow \exists Z \exists W\ \mathit{stock(X,Z,W)}\\
\sigma_5 &:& \mathit{stock\_portf(X,Y,Z)} \rightarrow \mathit{has\_stock(Y,X)}\\
\sigma_6 &:& \mathit{has\_stock(X,Y)} \rightarrow \exists Z\ \mathit{stock\_portf(Y,X,Z)}\\
\sigma_7 &:& \mathit{stock(X,Y,Z)} \rightarrow \exists V \exists W\ \mathit{stock\_portf(V,X,W)}\\
\sigma_8 &:& \mathit{stock(X,Y,Z)} \rightarrow \mathit{fin\_ins(X)}\\
\sigma_9 &:& \mathit{company(X,Y,Z)} \rightarrow \mathit{legal\_person(X)}\\
\delta_1 &:& \mathit{legal\_person(X,Y,Z), fin\_ins(X,V,W)
\rightarrow \bot}.
\end{array}
\]
}

\begin{figure*}
\caption{A (partial) rewriting for the Stock Exchange example.}
\label{fig:rewriting} \centering
\begin{tabular}{|l|}
\hline
\small $q^{[0]}(A,B,C) \leftarrow \mathit{fin\_ins}(A),\mathit{stock\_portf}(B,A,D),\mathit{company}(B,E,F),\mathit{list\_comp}(A,C),\mathit{fin\_idx}(C,G,H)$\\
\small $q^{[1]}(A,B,C) \leftarrow \mathit{fin\_ins}(A),\underline{\mathit{has\_stock}(A,B)},\mathit{company}(B,E,F),\mathit{list\_comp}(A,C),\mathit{fin\_idx}(C,G,H)$\\
\small $q^{[2]}(A,B,C) \leftarrow \mathit{fin\_ins}(A),\mathit{has\_stock}(A,B),\underline{\mathit{stock\_portf}(B,E,F)},\mathit{list\_comp}(A,C),\mathit{fin\_idx}(C,G,H)$\\
\small $q^{[3]}(A,B,C) \leftarrow \underline{\mathit{stock}(A,J,K)},\mathit{has\_stock}(A,B),\mathit{stock\_portf}(B,E,F),\mathit{list\_comp}(A,C),\mathit{fin\_idx}(C,G,H)$\\
$\ldots$\\
\hline
\end{tabular}
\end{figure*}

\vspace{-2mm}

The first four TGDs set the ``domain'' and the ``range'' of the
$\mathit{stock\_portf}$ and $\mathit{list\_comp}$ relations,
respectively. TGDs $\sigma_5$ and $\sigma_6$ assert that
$\mathit{stock\_portf}$ and $\mathit{has\_stock}$ are ``inverse
relations'', while $\sigma_7$ expresses that each stock must belong
to some stock portfolio. TGDs $\sigma_8$ and $\sigma_9$ model
taxonomic relationships such as the facts that each stock is a
financial instrument, and each company is a legal person. Finally,
the negative constraint $\delta_1$ (where $\bot$ denotes the truth
constant $\mathit{false}$) states that legal persons and financial
instruments are disjoint sets.

\vspace{-2mm}

Consider now the following conjunctive query $q$ asking for all the
triples $\tup{a,b,c}$, where $a$ is a financial instrument owned by
the company $b$ and listed on $c$:
\[
\begin{array}{rcl}
\mathit{q(A,B,C)} &\leftarrow& \mathit{fin\_ins}(A),
\mathit{stock\_portf}(B,A,D), \mathit{company}(B,E,F),\\
&&\mathit{list\_comp}(A,C), \mathit{fin\_idx}(C,G,H).
\end{array}
\]
Since $\dep = \{\sigma_1,\ldots,\sigma_9\}$ is a set of linear TGDs,
i.e., TGDs with single body-atom, query answering under $\dep$ is
FO-rewritable. Thus, it is possible to reformulate $\tup{\dep,q}$ to
a first-order query $q_{\Sigma}$ such that, for every database $D$,
$D \cup \dep \models q$ iff $D \models q_\dep$.
A naive rewriting procedure would use the TGDs of $\Sigma$ as
rewriting rules for the atoms in $q$ to generate all the CQs of the
perfect rewriting. Figure~\ref{fig:rewriting} shows a (partial)
rewriting for $q$, where the query obtained at the $i$-th step is
denoted as $q^{[i]}$, and the newly introduced atoms are underlined.
In particular, $q^{[0]}$ is the given query $q$, $q^{[1]}$ is
obtained from $q^{[0]}$ by using $\sigma_6$, $q^{[2]}$ is obtained
from $q^{[1]}$ by applying $\sigma_1$, and $q^{[3]}$ is obtained
from $q^{[2]}$ by using $\sigma_8$.

\vspace{-2mm}

The complete perfect rewriting contains more than 200 queries
executing more than 1000 joins. However, by exploiting the set of
constraints, it is possible to eliminate redundant atoms in the
generated queries, and thus reduce the number of the CQs in the
rewritten query.
For example, in the given query $q$ above it is possible to
eliminate the atom $\mathit{fin\_ins(A)}$ since, due to the
existence of the TGDs $\sigma_2$ and $\sigma_8$ in $\dep$, if the
atom $\mathit{stock\_portf}(B,A,D)$ is satisfied, then immediately
the atom $\mathit{fin\_ins(A)}$ is also satisfied.
Notice that by eliminating a redundant atom from a query, we also
eliminate all the atoms that are generated starting from it during
the rewriting process.
Moreover, due to the TGD $\sigma_3$, if the atom
$\mathit{list\_comp(A,C)}$ in $q$ is satisfied, then the atom
$\mathit{fin\_idx(C,G,H)}$ is also satisfied, and therefore can be
eliminated.
Finally, due to the TGD $\sigma_1$, if the atom
$\mathit{stock\_portf}(B,A,D)$ is satisfied, then the atom
$\mathit{company(B,E,F)}$ is also satisfied, and hence is redundant.
The query that has to be considered as input of the rewriting
process is therefore $q(A,B,C) \leftarrow
\mathit{stock\_portf}(B,A,D), \mathit{list\_comp}(A,C)$ that
produces a perfect rewriting containing the following two queries
executing only two joins:
\[
\begin{array}{rcl}
q(A,B,C) &\leftarrow& \mathit{list\_comp}(A,C), \mathit{stock\_portf}(B,A,D)\\
q(A,B,C) &\leftarrow& \mathit{list\_comp}(A,C),
\mathit{has\_stock}(A,B).
\end{array}
\]

\vspace{-2mm}

\paragraph{Contributions and Roadmap.} After a review of previous work on
ontology based data access in the next section, and some formal
definitions and preliminaries in Section~\ref{sec:preliminaries}, we
present a short overview of the {\plusminus} family in
Section~\ref{sec:datalogpm-family}.
We then proceed with new research results. In
Section~\ref{sec:datalog-odba}, we propose a new rewriting algorithm
that improves the one stated in~\cite{WCGP10} by substantially
reducing the number of redundant queries in the perfect rewriting.
In Section~\ref{sec:optimization}, we  present a polynomial-time
optimization strategy based on the early-pruning of redundant atoms
produced during the rewriting process.
An implementation and experimental evaluation of the new method is
discussed in Section~\ref{sec:experiments}. We also discuss the
relationship between our optimization technique and optimal query
minimization algorithms such as the \emph{chase \& back-chase}
algorithm~\cite{DPT06}. We conclude with a brief outlook on further
research.

\section{Ontology Based Data Access}
\label{sec:related}

Answering queries under constraints and the related optimization
techniques are important topics in data management beyond the
obvious research interest. In particular, they are profitable
opportunities for companies that need to deliver efficient and
effective data management solutions to their customers. This trend
is becoming even more evident as a plethora of robust systems and
APIs for Semantic Web data management proposed in the recent years.
These systems span from open-source solutions such as
Virtuoso\footnote{http://virtuoso.openlinksw.com/},
Sesame\footnote{http://www.openrdf.org/}, RDFSuite~\cite{alexaki01},
KAON\footnote{http://kaon.semanticweb.org/} and
Jena\footnote{http://jena.sourceforge.net/}, to commercial
implementations such as the semantic extensions implemented in
Oracle Database 11g R2~\cite{Ch05} and
BigOWLLim\footnote{http://www.ontotext.com/owlim/}.
In this Section we briefly analyze the systems providing
rewriting-based access to databases under ontological constraints,
and we highlight some crucial points that we want to address in this
work.

\vspace{-2mm}

We first present the class of constraints identified by the members
of the DL-Lite family~\cite{CDLLLR07}, namely,
DL-Lite$_{\mathcal{A}}$, DL-Lite$_{\mathcal{F}}$, and
DL-Lite$_{\mathcal{R}}$, underlying the W3C OWL-QL profile of the
OWL language. These constraints correspond to unary and binary
\emph{inclusion dependencies} combined with a restricted form of
\emph{key constraints}. In order to perform query answering under
this class of constraints, a rewriting algorithm, introduced
in~\cite{CDLLLR07} and implemented in the QuOnto system,
reformulates the given query into unions of conjunctive queries. The
size of the reformulated query is unnecessarily large due to a
number of reasons. In the first place, \emph{(i)} basic optimization
techniques such as the identification of the connected components in
the body of the input query, or the computation of any form of query
decomposition~\cite{GLS99}, are not applied. Moreover, \emph{(ii)}
the fact that the given set of constraints can be used to identify
existential joins in the reformulated query which can be eliminated
is not exploited. Finally, \emph{(iii)} the factorization step
(which is needed to guarantee completeness) is applied exhaustively,
and as a result many superfluous queries are generated.

\vspace{-2mm}

Per\'ez-Urbina et al.~\cite{PMH09} proposed an alternative
resolution-based rewriting algorithm, implemented in the Requiem
system, that addressed the issue of the useless factorizations (and
therefore of the redundant queries generated due to this weakness)
by directly handling existential quantification through proper
functional terms. The algorithm has then been extended to more
expressive DL languages~\cite{PMH09}. In this case the output of the
rewriting is a Datalog program.

\vspace{-2mm}

Rosati et al.~\cite{AR10} recently proposed a very sophisticated
rewriting technique, implemented in the Presto system, that
addresses some of the issues described above. In particular,
\emph{(i)} the unnecessary existential joins are eliminated by
resorting to the concept of \emph{most-general subsumees}, which
also avoids the unnecessary factorizations, and \emph{(ii)} the
connectivity of the given query is checked before executing the
algorithm; in case the query is not connected, Presto splits the
query in connected components and rewrites them separately.
Notice that Presto produces a non-recursive Datalog program, and not
a union of conjunctive queries. This allows the ``hiding'' of the
exponential blow-up inside the rules instead of generating
explicitly the disjunctive normal form. The final rewriting is
exponential only in the number of non-eliminable existential joins,
but not in the size of the input query.


\vspace{-2mm}

The approaches presented above have been proven very effective when
applied to very particular classes of description logic constraints.
Following a more general approach for ontological query answering,
Cal\`i et al.~\cite{WCGP10} presented a backward-chaining rewriting
algorithm which is able to deal with arbitrary sets of TGDs,
providing that the class of TGDs under consideration satisfies
suitable syntactic restrictions that guarantee the termination of
the algorithm. However, this algorithm is inspired by the original
QuOnto algorithm and inherits all its drawbacks.

\vspace{-2mm}

Despite the possibly exponential number of queries to be constructed, we know that all such queries are
independent from each other, and thus they can be easily executed in
parallel threads and distributed on multiple processors.
Notice that a non-recursive Datalog program is not equally easy to distribute.
Moreover, the optimizations implemented in current DBMS systems for
(unions of) conjunctive queries are much more advanced than those
implemented for the positive existential first-order queries
resulting from the translation of a non-recursive Datalog program
into a concrete query language such as SQL. It is clear that a
trade-off between these two approaches must be found in order to
exploit as much as possible the current optimization techniques,
while keeping the size of the rewriting reasonably small in order to
make the execution of it feasible in practice.

\vspace{-2mm}

A related research field is that of query minimization~\cite{CM77},
in particular, in presence of views and
constraints~\cite{Hal01,DPT06}. Given a conjunctive query $q$, and a
set of constraints $\Sigma$, the goal is to find all the minimal
equivalent reformulations of $q$ under the constraints of $\Sigma$.
The most interesting approach in this respect is the chase \&
back-chase algorithm (C\&B)~\cite{DPT06}, implemented in the MARS
system~\cite{DeTa03}. The algorithm freezes the atoms of $\body{q}$
and, by considering them as a database $D_q$, applies the following
two steps.
During the {\em chase-step}, the chase of $D_q$ w.r.t. $\dep$ is
constructed, and then the atoms of $\chase{D_q}{\Sigma}$ are
considered as the body-atoms of a query $q_{u}$, called the {\em
universal plan}.
The {\em back-chase step} considers all the possible subsets of the
atoms of $\body{q_u}$, starting from those with a single-atom, which
are then considered as the body of a query $q'$. Whenever there
exists a containment mapping from $\body{q_u}$ to
$\chase{D_{q'}}{\dep}$, where $D_{q'}$ is the database obtained by
freezing $\body{q'}$, then $q^{\prime}$ is an equivalent
reformulation of $q$.
Moreover, every time an equivalent reformulation $q^{\prime}$ is
found, the back-chase does not consider any of the supersets of the
atoms of $\body{q^{\prime}}$ because they will be automatically
implied by the atoms of $q^{\prime}$, and therefore the produced
query would be redundant. This particular exploration strategy
guarantees the minimality of the reformulations.
A non-negligible drawback of this approach is the fact that we need
to compute the chase of $D_q$ w.r.t. $\dep$, and also the chase for
the (exponentially many) databases $D_{q'}$ w.r.t $\dep$. Clearly,
this makes the procedure computationally expensive.

\section{Preliminaries}
\label{sec:preliminaries}

In this section we recall some basics on relational databases,
conjunctive queries, tuple-generating dependencies, and the chase
procedure.

\subsection{Relational Databases and Conjunctive Queries}
\label{subsec:db-cq}

Consider two pairwise disjoint (infinite) sets of symbols $\Delta_c$
and $\Delta_z$ such that: $\Delta_c$ is a set of \emph{constants}
(which constitutes the domain of a database), and $\Delta_z$ is a
set of \emph{labeled nulls} (used as placeholders for unknown
values). Different constants represent different values
(\emph{unique name assumption}), while different nulls may represent
the same value.
Throughout the paper, we denote by $\insX$ sequences of variables
$X_1,\ldots,X_k$, where $k \geq 0$, and by $[n]$ the set
$\{1,\ldots,n\}$, for any $n \geq 1$.

\vspace{-2mm}

A \emph{relational schema} $\mathcal{R}$ (or simply \emph{schema})
is a set of \emph{relational symbols} (or \emph{predicate symbols}),
each with its associated arity. A \emph{position} $r[i]$ (or
$\tup{r,i}$) is identified by a predicate $r \in \R$ and its $i$-th
argument. A \emph{term} $t$ is a constant, labeled null, or
variable. An \emph{atomic formula} (or simply \emph{atom}) has the
form $r(t_1, \ldots, t_n)$, where $r \in \R$ has arity $n$, and
$t_1,\ldots,t_n$ are terms. Conjunctions of atoms are often
identified with the sets of their atoms.

\vspace{-2mm}

A \emph{substitution} from one set of symbols $S_{1}$ to another set
of symbols $S_{2}$ is a function $h : S_{1} \rightarrow S_{2}$.
A \emph{homomorphism} from a set of atoms $A_{1}$ to a set of
atoms $A_{2}$, both over the same schema $\R$, is a substitution $h$
from the set of terms of $A_1$ to the set of terms of $A_2$ such
that: \textit{(i)} if $t \in \Delta_c$, then $h(t) = t$, and
\textit{(ii)} if $r(t_{1}, \ldots, t_{n})$ is in $A_{1}$, then
$h(r(t_{1}, \ldots, t_{n})) = r(h(t_{1}), \ldots, h(t_{n}))$ is in
$A_{2}$. The notion of homomorphism naturally extends to
conjunctions of atoms.

\vspace{-2mm}

A \emph{relational instance} (or simply \emph{instance}) $I$ for a
schema $\mathcal{R}$ is a (possibly infinite) set of atoms of the
form $r({\bf t})$, where $r \in \mathcal{R}$ has arity $n$ and ${\bf
t} \in (\Delta_{c} \cup \Delta_{z})^{n}$. A \emph{database} is a
finite relational instance.
A \emph{conjunctive query} (CQ) $q$ of arity $n$ over a schema $\R$
is a formula of the form $q(\insX) \leftarrow \phi(\insX,\insY)$,
where $\phi(\insX,\insY)$ is a conjunction of atoms over $\R$, and
$q$ is an $n$-ary predicate. $\phi(\insX,\insY)$ is called the {\em
body} of $q$, denoted as $body(q)$, and $q(\insX)$ is the head of
$q$, denoted as $head(q)$. A {\em Boolean conjunctive query} (BCQ)
is a CQ of arity zero.
The {\em answer} to a CQ $q$ of arity $n$ over an instance $I$,
denoted as $q(I)$, is the set of all $n$-tuples $\tuple{t} \in
(\Delta_c)^{n}$ for which there exists a homomorphism $h : \insX
\cup \insY \rightarrow \Delta_c \cup \Delta_z$ such that
$h(\phi(\insX,\insY)) \subseteq I$ and $h(\insX) = \tuple{t}$. A BCQ
has only the empty tuple $\langle \rangle$ as possible answer, in
which case it is said that has positive answer. Formally, a BCQ has
\emph{positive} answer over $I$, denoted as $I \models q$, iff
$\langle \rangle \in q(I)$.
A \emph{union of CQs} (UCQ) $Q$ of arity $n$ is a set of CQs, where
each $q \in Q$ has the same arity $n$ and uses the same predicate
symbol in the head. The answer to $Q$ over an instance $I$, denoted
as $Q(I)$, is defined as the set of tuples
$\{\tuple{t}~|~\textrm{there~exists~} q \in Q \textrm{~such~that~}
\tuple{t} \in q(I)\}$.

\subsection{Tuple-Generating Dependencies}
\label{subsec:tgds}

A {\em tuple-generating dependency} (TGD) $\sigma$ over a schema
$\R$ is a first-order formula $\forall{\bf X} \forall{\bf Y}
\phi({\bf X}, {\bf Y}) \rightarrow \exists \insZ \, \psi({\bf X},
{\bf Z})$, where $\phi(\insX,\insY)$ and $\psi(\insX,\insZ)$ are
conjunctions of atoms over $\R$, called the {\em body} and the {\em
head} of $\sigma$, denoted as $\body{\sigma}$ and $\head{\sigma}$,
respectively. Henceforth, to avoid notational clutter, we will omit
the universal quantifiers in TGDs.
Such $\sigma$ is satisfied by an instance $I$ for $\R$ iff, whenever
there exists a homomorphism $h$ such that $h(\phi(\insX,\insY))
\subseteq I$, there exists an extension $h'$ of $h$ (i.e., $h'
\supseteq h$) such that $h'(\psi(\insX,\insZ)) \subseteq I$.

\vspace{-2mm}

We now define the notion of \emph{query answering} under TGDs. Given
a database $D$ for $\R$, and a set $\dep$ of TGDs over $\R$, the
\emph{models} of $D$ w.r.t.~$\dep$, denoted as $\mods{D}{\dep}$, is
the set of all instances $I$ such that $I \models D \cup \dep$,
which means that $I \supseteq D$ and $I$ satisfies $\dep$.  The
\emph{answer} to a CQ $q$ w.r.t.~$D$ and $\dep$, denoted as
$\ans{q}{D}{\dep}$, is the set $\{\tuple{t}~|~\tuple{t} \in q(I)
{\rm~for~each~} I \in \mods{D}{\dep}\}$.  The \emph{answer} to a BCQ
$q$ w.r.t.~$D$ and $\dep$ is \emph{positive}, denoted as $D \cup
\dep \models q$, iff $\ans{q}{D}{\dep} \neq \emptyset$. Note that
query answering under general TGDs is undecidable \cite{BeVa81},
even when the schema and the set of TGDs are fixed \cite{CGK08}.
We recall that the two problems of answering CQs and BCQs under TGDs
are equivalent~\cite{CM77,DNR08}. Roughly speaking, we can enumerate
the polynomially many tuples of constants which are possible answers
to $q$, and then, instead of answering the given query $q$, we
answer the polynomially many BCQs that we obtain by replacing the
variables in the body of $q$ with the appropriate constants. A
certain tuple $\tuple{t}$ of constants is in the answer of $q$ iff
the answer to the BCQ that we obtain from $\tuple{t}$ is positive.
Henceforth, we thus focus only on the BCQ answering problem.

\subsection{The TGD Chase}
\label{subsec:tgd-chase}

The \emph{chase procedure} (or simply \emph{chase}) is a fundamental
algorithmic tool introduced for checking implication of
dependencies~\cite{MMS79}, and later for checking query
containment~\cite{JoKl84}. Informally, the chase is a process of
repairing a database w.r.t.~a set of dependencies so that the
resulted database satisfies the dependencies. We shall use the term
chase interchangeably for both the procedure and its result. The
chase works on an instance through the so-called TGD \emph{chase
rule}.

\vspace{-2mm}

\textsc{TGD Chase Rule:} Consider a database $D$ for a schema $\R$,
and a TGD $\sigma : \phi(\insX,\insY) \rightarrow \exists
\insZ\,\psi(\insX,\insZ)$ over $\R$. If $\sigma$ is {\em applicable}
to $D$, i.e., there exists a homomorphism $h$ such that
$h(\phi(\insX,\insY)) \subseteq D$ then: \emph{(i)} define $h'
\supseteq h$ such that $h'(Z_{i}) = z_{i}$, for each $Z_{i} \in
\insZ$, where $z_{i} \in \Delta_z$ is a ``fresh'' labeled null not
introduced before, and \emph{(ii)} add to $D$ the set of atoms in
$h'(\psi(\insX,\insZ))$, if not already in $D$.

\vspace{-2mm}

Given a database $D$ and a set of TGDs $\dep$, the chase algorithm
for $D$ and $\dep$ consists of an exhaustive application of the TGD
chase rule in a breadth-first fashion, which leads as result to a
(possibly infinite) chase for $D$ and $\dep$, denoted as
$\chase{D}{\dep}$. For the formal definition of the chase algorithm
we refer the reader to~\cite{CaGL09}.

\vspace{-2mm}

The (possibly infinite) chase for $D$ and $\dep$ is a
\emph{universal model} of $D$ w.r.t.~$\dep$, i.e., for each instance
$I \in \mods{D}{\dep}$, there exists a homomorphism from
$\chase{D}{\dep}$ to $I$~\cite{DNR08,FKMP05}. Using this fact it can
be shown that $D \cup \dep \models q$ iff $\chase{D}{\dep} \models
q$, for every BCQ $q$.

\section{The {\plusminus} Family}
\label{sec:datalogpm-family}

In this section we present the main {\plusminus} languages under which
query answering is decidable, and (almost in all cases) also
tractable in data complexity.

\subsection{Decidability Paradigms}
\label{subsec:decidability-paradigms}

We first discuss the three main paradigms for ensuring decidability
of query answering, namely, chase termination, guardedness and
stickiness.

\vspace{-2mm}

\paragraph{Chase Termination.} In this case the chase always terminates and
produces a finite universal model $U$. Thus, given a query we just
need to evaluate it over the finite database $U$.
The most notable syntactic restriction of TGDs guaranteeing chase
termination is \emph{weak-acyclicity}, which is defined by means of
a graph-based condition, for which we refer the reader to the
landmark paper~\cite{FKMP05}. Roughly speaking, in the chase
constructed under a weakly-acyclic set of TGDs over a schema
$\mathcal{R}$, only a finite number of distinct values can appear at
any position of $\mathcal{R}$, and thus after finitely many steps
the chase procedure terminates. It is known that query answering
under a weakly-acyclic set of TGDs is \textsc{ptime}-complete
\cite{FKMP05} and 2\textsc{exptime}-complete~\cite{CaGP10} in data
and combined complexity, respectively.
More general syntactic restrictions that guarantee chase termination
were proposed in~\cite{DNR08} and~\cite{Mar09}.

\vspace{-2mm}

\paragraph{Guardedness.} \emph{Guarded} TGDs, introduced and studied
in~\cite{CGK08}, have an atom in their body, called the
\emph{guard}, that contains all the universally quantified
variables. For example, the TGD $r(X,Y ),s(X,Y,Z) \rightarrow
\exists W s(Z,X,W)$ is guarded via the guard atom $s(X,Y,Z)$, while
the TGD $r(X,Y),r(Y,Z) \rightarrow r(X,Z)$ is not.
Decidability of query answering follows from the fact that the chase
constructed under a set of guarded TGDs has the bounded treewidth
property, i.e., is a ``tree-like'' structure.
The data and combined complexity of query answering under a set of
guarded TGDs is \textsc{ptime}-complete~\cite{CGL09} and
2\textsc{exptime}-complete~\cite{CGK08}, respectively.

\vspace{-2mm}

\emph{Linear} TGDs, proposed in~\cite{CGL09}, is a FO-rewritable
variant of guarded TGDs. A TGD is linear iff it contains only one
atom in its body. Obviously a linear TGD is trivially guarded since
the singleton body-atom is automatically a guard. Linear
TGDs are more expressive than the well-known class of inclusion
dependencies.
Query answering under linear TGDs
is in the highly tractable class \textsc{ac}$_0$ in data
complexity~\cite{CGL09}. The same problem is \textsc{pspace}-complete
in combined complexity; this result is immediately implied by
results in~\cite{JoKl84}.

\vspace{-2mm}

An expressive class, which forms a generalization of guarded TGDs,
is the class of \emph{weakly-guarded} sets of TGDs introduced
in~\cite{CGK08}.
Intuitively speaking, a set $\Sigma$ of TGDs is weakly-guarded iff
in the body of each TGD of $\Sigma$ there exists an atom, called the
\emph{weak-guard}, that contains all the universally quantified
variables that appear only at positions where a ``fresh'' null of
$\Delta_z$ can appear during the construction of the chase.
Query answering under a weakly-guarded set of TGDs is
\textsc{exptime}-complete~\cite{CGK08} and
2\textsc{exptime}-complete~\cite{CGK08} in data and combined
complexity, respectively.

\vspace{-2mm}

\paragraph{Stickiness.} In this paragraph we present a \plusminus~language
(and its extensions), which hinges on a paradigm that is very
different from guardedness. \emph{Sticky} sets of TGDs are defined
formally by an efficiently testable condition involving
variable-marking~\cite{CGP10}. In what follows we just give an
intuitive definition of this class. For every database $D$, assume
that during the construction of chase of $D$ under a sticky set of
TGDs, we apply a TGD $\sigma \in \Sigma$ that has a variable $V$
appearing more than once in its body; assume also that $V$ maps (via
homomorphism) on the symbol $z$, and that by virtue of this
application the atom $\atom{a}$ is introduced. In this case, for
each atom $\atom{b}$ in $\body{\sigma}$, we say that $\atom{a}$ is
\emph{derived} from $\atom{b}$. Then, we have that $z$ appears in
$\atom{a}$ and in all atoms resulting from some chase derivation
sequence starting from $\atom{a}$, ``sticking'' to them (hence the
name ``sticky'' sets of TGDs). Interestingly, sticky sets of TGDs
are FO-rewritable, and thus query answering is feasible in
\textsc{ac}$_0$ in data complexity~\cite{CGP10}. Combined complexity
of query answering is known to be
\textsc{exptime}-complete~\cite{CGP10}.

\vspace{-2mm}

In~\cite{CaGP10} the FO-rewritable class of \emph{sticky-join} sets
of TGDs, that captures both linear TGDs and sticky sets of TGDs, is
introduced. Similarly to sticky sets of TGDs, sticky-join sets are
defined formally by a testable condition based on variable-marking.
However, this variable-marking procedure is more sophisticated than
the one used for sticky sets, and due to this fact the problem of
identifying sticky-join sets of TGDs is harder than the one of
identifying sticky sets. In particular, given a set $\Sigma$ of
TGDs, we can decide in \textsc{ptime} whether $\Sigma$ is sticky,
while the problem whether $\Sigma$ is sticky-join is
\textsc{pspace}-complete.
Note that the data and combined complexity of query answering under
sticky and sticky-join sets of TGDs coincide.

\subsection{Additional Features}
\label{subsec:additional-features}

In this subsection we briefly discuss how the languages presented
above can be combined with negative constraints and key
dependencies, without altering the complexity of query answering.

\vspace{-2mm}

\paragraph{Negative Constraints.} A \emph{negative constraint} (NC) $\nu$
over a schema $\R$ is a first-order formula $\forall {\bf X} \,
\phi({\bf X}) \rightarrow \bot$, where $\bot$ denotes the truth
constant {\em false}. NCs are vital when representing ontologies
(see, e.g., \cite{CGL09,CGP10}), as well as conceptual schemas such
as Entity-Relationship diagrams~(see, e.g., \cite{CGP09,CGP10-ER}).
With NCs we can assert, for example, that students and professors
are disjoint sets: $\forall X
\mathit{student}(X),\mathit{professor}(X) \ra \bot$. Also, we can
state that a student cannot be the leader of a research group:
$\forall X \forall Y \mathit{student}(X),\mathit{leads}(X,Y) \ra
\bot$.

\vspace{-2mm}

It is known that checking NCs is tantamount to query
answering~\cite{CGL09}. In particular, given an instance $I$, a set
$\ndep$ of NCs, and a set $\dep$ of TGDs, for each NC $\nu$ of the
form $\forall \insX \, \phi(\insX) \rightarrow \bot$, we answer the
BCQ $q_\nu() \la \phi(\insX)$. If at least one of such queries
answers positively, then $I \cup \dep \cup \ndep \models \bot$
(i.e., the theory is inconsistent), and therefore $I \cup \dep \cup
\ndep \models q$, for every BCQ $q$; otherwise, given a BCQ $q$, we
have $I \cup \dep \cup \ndep \models q$ iff $I \cup \dep \models q$,
i.e., we can answer $q$ by ignoring the set of NCs.

\vspace{-2mm}

\paragraph{Key Dependencies.} It is well-known that the interaction of
general TGDs and key dependencies (KDs) leads to undecidability of
query answering~\cite{ChVa85}; we assume that the reader is familiar
with the notion of KD (see, e.g., \cite{AHV95}). Thus, the classes
of TGDs presented above cannot be combined arbitrarily with KDs.
Suitable syntactic restrictions are needed in order to ensure
decidability of query answering.

\vspace{-2mm}

A crucial concept towards this direction is
separability~\cite{CLR03}, which formulates a controlled interaction
of TGDs and KDs. Formally speaking, a set $\dep = \tdep \cup \kdep$
over a schema $\R$, where $\tdep$ and $\kdep$ are sets of TGDs and
KDs, respectively, is \emph{separable} iff for every instance $I$
for $\R$, either $I$ violates $\kdep$, or for every BCQ $q$ over
$\R$, $I \cup \dep \models q$ iff $I \cup \tdep \models q$.
Notice that separability is a semantic notion. A sufficient
syntactic criterion for separability of TGDs and KDs is given
in~\cite{CGL09}; TGDs and KDs satisfying the criterion are called
\emph{non-conflicting}.

\vspace{-2mm}

Obviously, in case of non-conflicting sets of TGDs and KDs, we just
need to perform a preliminary check whether the given instance
satisfies the KDs, and if this is the case, then we eliminate them,
and proceed by considering only the set of TGDs.
This preliminary check can be done using negative constraints. For
example, to check whether the KD $\mathit{key}(r) = \{1\}$, stating
that the first attribute of the binary relation $r$ is a key
attribute, is satisfied by the database $D$, we just need to check
whether the database $D_{\neq}$ obtained by adding to $D$ the set of
atoms $\{\mathit{neq}(a,b)~|~a \neq b, \textrm{~and~} a,b
\textrm{~are~constants~occurring~in~} D\}$, where $\mathit{neq}$ is
an auxiliary predicate, satisfies the negative constraint
$r(X,Y),r(X,Z),\mathit{neq}(Y,Z) \ra \bot$. The atom
$\mathit{neq}(a,b)$ implies that $a$ and $b$ are different
constants.
Since, as already mentioned, checking NCs is tantamount to query
answering, we immediately get that the complexity of query answering
under non-conflicting sets of TGDs and KDs is the same as in the
case of TGDs only.

\vspace{-2mm}

Interestingly, by combining non-conflicting linear (or sticky) sets
of TGDs and KDs with NCs, we get strictly more expressive formalisms
than the most widely-adopted tractable ontology languages, in
particular DL-Lite$_{\A}$, DL-Lite$_{\F}$ and DL-Lite$_{\R}$,
without loosing FO-rewritability, and consequently high tractability
of query answering in data complexity. For more details, we refer
the interested reader to~\cite{CGL09,CGP10}.

\section{{\plusminus} for OBDA}
\label{sec:datalog-odba}

In this section we consider the problem of BCQ answering under the
FO-rewritable members of the {\plusminus} family, namely, linear,
sticky and sticky-join sets of TGDs. Given a BCQ $q$ and a set
$\dep$ of TGDs, the actual computation of the rewriting is done by
applying a backward-chaining resolution procedure using the rules of
$\dep$ as rewriting rules. Our algorithm optimizes the algorithm
presented in~\cite{WCGP10} by greatly reducing the number of BCQs in
the rewriting, and therefore improves the overall performance of
query answering. Before going into the details of the rewriting
algorithm, we first give some useful notions.

\vspace{-2mm}

A set of atoms $A = \{\atom{a}_1,\ldots,\atom{a}_n\}$, where $n
\geqslant 2$, \emph{unifies} if there exists a substitution
$\gamma$, called \emph{unifier} for $A$, such that
$\gamma(\atom{a}_1) = \ldots = \gamma(\atom{a}_n)$. A \emph{most
general unifier (MGU)} for $A$ is a unifier for $A$, denoted as
$\gamma_{A}$, such that for each other unifier $\gamma$ for $A$,
there exists a substitution $\gamma^{\prime}$ such that $\gamma =
\gamma^{\prime} \circ \gamma_{A}$. Notice that if a set of atoms
unify, then there exists a MGU. Furthermore, the MGU for a set of
atoms is unique (modulo variable renaming). The MGU for a singleton
set $\{\atom{a}\}$ is defined as the identity substitution on the
set of terms that occur in $\atom{a}$.

\vspace{-2mm}

Let us now give some auxiliary results which will allow us to
simplify our later technical definitions and proofs. The first such
lemma states that we can restrict our attention on TGDs that have
only one head-atom.

\begin{lemma}\label{lem:one-head-atom}
BCQ answering under (general) TGDs and BCQ answering under TGDs with
just one head-atom are \textsc{logspace}-equivalent problems.
\end{lemma}

\begin{proof}
It suffices to show that BCQ answering under (general) TGDs can be
reduced in \textsc{logspace} to BCQ answering under TGDs with just
one head-atom.
Consider a BCQ $q$ over a schema $\R$, a database $D$ for $\R$, and
a set $\dep$ of TGDs over $\R$. We construct $\dep'$ from $\dep$ by
applying the following procedure. For each TGD $\sigma \in \dep$,
where $\head{\sigma} = \{\atom{a}_1,\ldots,\atom{a}_k\}$ and $\insX$
is the set of variables that occur in $\head{\sigma}$, replace
$\sigma$ with the following set of TGDs:
\[
\begin{array}{rcl}
\body{\sigma} &\ra& r_{\sigma}(\insX),\\
r_{\sigma}(\insX) &\ra& \atom{a}_1,\\
r_{\sigma}(\insX) &\ra& \atom{a}_2,\\
&\vdots&\\
r_{\sigma}(\insX) &\ra& \atom{a}_k,
\end{array}
\]
where $r_{\sigma}$ is an auxiliary predicate not occurring in $\R$
having the same arity as the number of variables in $\insX$. It is
not difficult to see that the above construction is feasible in
\textsc{logspace}.
By construction, except for the atoms with an auxiliary predicate,
$\chase{D}{\dep}$ and $\chase{D}{\dep'}$ coincide. The auxiliary
predicates, being introduced only during the above transformation,
do not match any predicate symbol in $q$, and hence $\chase{D}{\dep}
\models q$ iff $\chase{D}{\dep'} \models q$, or, equivalently, $D
\cup \dep \models q$ iff $D \cup \dep' \models q'$. \qed
\end{proof}

The next lemma implies that we can restrict our attention on TGDs
that have only one existentially quantified variable which occurs
only once.

\begin{lemma}\label{lem:one-exist-variable}
BCQ answering under (general) TGDs and BCQ answering under TGDs with
at most one existentially quantified variable that occurs only once
are \textsc{logspace}-equivalent problems.
\end{lemma}

\begin{proof}
It suffices to show that BCQ answering under (general) TGDs can be
reduced in \textsc{logspace} to BCQ answering under TGDs that have
at most one existentially quantified variable which occurs only
once.
Consider a BCQ $q$ over a schema $\R$, a database $D$ for $\R$, and
a set $\dep$ of TGDs over $\R$. We construct $\dep'$ from $\dep$ by
applying the following procedure. For each TGD $\sigma \in \dep$,
where $\{X_1,\ldots,X_n\}$, for $n \geqslant 1$, is the set of
variables that occur both in $\body{\sigma}$ and $\head{\sigma}$,
and $\{Z_1,\ldots,Z_m\}$, for $m
> 1$, is the set of the existentially quantified variables of
$\sigma$, replace $\sigma$ with the following set of TGDs:
\[
\begin{array}{rcl}
\body{\sigma} &\ra& \exists Z_1 \,
r_{\sigma}^{1}(X_1,\ldots,X_n,Z_1),\\
r_{\sigma}^{1}(X_1,\ldots,X_n,Z_1) &\ra& \exists Z_2 \,
r_{\sigma}^{2}(X_1,\ldots,X_n,Z_1,Z_2),\\
&\vdots&\\
r_{\sigma}^{m-1}(X_1,\ldots,X_n,Z_1,\ldots,Z_{m-1}) &\ra& \exists
Z_m \, r_{\sigma}^{m}(X_1,\ldots,X_n,Z_1,\ldots,Z_m),\\
r_{\sigma}^{m}(X_1,\ldots,X_n,Z_1,\ldots,Z_m) &\ra& \head{\sigma},
\end{array}
\]
where $r_{\sigma}^{i}$ is an auxiliary predicate of arity $n+i$, for
each $i \in [m]$. It is easy to see that the above procedure can be
carried out in \textsc{logspace}.
By construction, except for the atoms with an auxiliary predicate,
$\chase{D}{\dep}$ and $\chase{D}{\dep'}$ are the same (modulo
bijective variable renaming). The auxiliary predicates, being
introduced only during the above construction, do not match any
predicate symbol in $q$, and hence $\chase{D}{\dep} \models q$ iff
$\chase{D}{\dep'} \models q$, or, equivalently, $D \cup \dep \models
q$ iff $D \cup \dep' \models q$. \qed
\end{proof}

Since the transformations given above preserve the syntactic
condition of linear, sticky and sticky-join sets of TGDs, henceforth
we assume w.l.o.g. that every TGD has just one atom in its head
which contains only one existentially quantified variable that
occurs only once. In the rest of the paper, for notational
convenience, given a TGD $\sigma$, we denote by $\pi_{\sigma}$ the
position in $\head{\sigma}$ at which the existentially quantified
variable occurs.

\vspace{-2mm}

We now give the notion of \emph{applicability} of a TGD to a set of
body-atoms of a query. Let us assume w.l.o.g that the variables that
appear in the query, and those that appear in the TGD, constitute
two disjoint sets.
Given a BCQ $q$, a variable is called \emph{shared} in $q$ if it
occurs more than once in $\body{q}$. Notice that in the case of
(non-Boolean) CQs, a variable is shared in $q$ if it occurs more
than once in $q$ (considering also the head of $q$ and not just its
body).

\begin{definition}[Applicability]
Consider a BCQ $q$ over a schema $\R$, and a TGD $\sigma$ over $\R$.
Given a set of atoms $A \subseteq \body{q}$ that unifies, we say
that $\sigma$ is \emph{applicable} to $A$ if the following
conditions are satisfied: (i) the set $A \cup \{\head{\sigma}\}$
unifies, and (ii) for each $\atom{a} \in A$, if the term at position
$\pi$ in $\atom{a}$ is either a constant or a shared variable in
$q$, then $\pi \neq \pi_{\sigma}$. \label{def:applicability}
\end{definition}

Let us now introduce the notion of \emph{factorizability} which, as
we explain below, makes one of the main differences between our
algorithm and the one presented in~\cite{WCGP10}, due to which a
perfect rewriting with less BCQs is obtained.

\begin{definition}[Factorizability]
Consider a BCQ $q$ over a schema $\R$, and a TGD $\sigma$ over $\R$
which contains an existentially quantified variable. A set of atoms
$A \subseteq \body{q}$, where $|A| \geqslant 2$, that unifies is
\emph{factorizable} w.r.t.~$\sigma$ if there exists a variable $V$
that occurs in every atom of $S$ only at position $\pi_\sigma$, and
also $V$ does not occur in $\body{q} \setminus S$.
\label{def:factorisability}
\end{definition}

It is important to clarify that in the case of (non-Boolean) CQs,
the notion of factorizability is defined as above, except that the
variable $V$ does not occur in $(\{\head{\sigma}\} \cup
\body{\sigma}) \setminus S$.

\begin{example}[Factorization]
{\rm Consider the BCQs
\[
\begin{array}{rcl}
q_1 &:& q()\, \leftarrow\, \underbrace{t(A,B,C),t(A,E,C)}_{S_1}\\
q_2 &:& q()\, \leftarrow\, s(C),\underbrace{t(A,B,C),t(A,E,C)}_{S_2}\\
q_3 &:& q()\, \leftarrow\, \underbrace{t(A,B,C),t(A,C,C)}_{S_3}
\end{array}
\]
and the TGD $\sigma : s(X),r(X,Y)\, \rightarrow\, \exists Z \,
t(X,Y,Z).$ Clearly, $S_1$ is factorizable w.r.t.~$\sigma$ since the
substitution $\{E \ra B\}$ is a unifier for $S_1$, and also $C$
appears in both atoms of $S_1$ only at position $\pi_{\sigma}$. The
factorization results in the query $q() \leftarrow t(A,B,C)$; notice
that $\sigma$ is not applicable to $S_1$, but it is applicable to
$\{t(A,B,C)\}$.
On the contrary, despite the fact that $S_2$ unifies, it is not
factorizable w.r.t.~$\sigma$ since $C$ occurs also in $\body{q_2}
\setminus S_2$.
Finally, even if $S_3$ unifies, it is not factorizable
w.r.t.~$\sigma$ since $C$ appears in $S_3$, not only at position
$\pi_{\sigma}$, but also at position $t[2]$.}
\label{exm:factorisation}
\end{example}

We are now ready to describe the algorithm {\sf TGD-rewrite},
depicted in Algorithm~\ref{alg:tgd-rewrite}, which is based on the
rewriting algorithm presented in~\cite{WCGP10}. The perfect
rewriting of a BCQ $q$ w.r.t. a set of TGDs $\Sigma$ is computed by
exhaustively applying (i.e., until a fixpoint is reached) two steps:
\emph{factorization} and \emph{rewriting}.

\begin{algorithm}[hbt]
\caption{The algorithm {\sf TGD-rewrite} \label{alg:tgd-rewrite}}
\small
    \KwIn{a BCQ $q$ over a schema $\R$, a set $\Sigma$ of TGDs over $\R$}
    \KwOut{the FO-rewriting $Q_{\textsc{fin}}$ of $q$ w.r.t. $\Sigma$}
    $Q_{\textsc{rew}} := \{\langle q,1 \rangle\}$\;
    \Repeat{$Q_{\textsc{temp}} = Q_{\textsc{rew}}$}{
        $Q_{\textsc{temp}} := Q_{\textsc{rew}}$\;
        \ForEach{$\{\langle q,x \rangle\} \in Q_{\textsc{temp}}$, {\em where} $x \in \{0,1\}$,}{
            \tcc{\textrm{factorization step}}
            \ForEach{$\sigma \in \Sigma$}{
                $q^{\prime} := \mathit{factorize}(q,\sigma)$\;
                \If{$\mathit{notExists}(\langle q^{\prime},y \rangle, Q_{\textsc{rew}})$, {\em where} $y \in \{0,1\}$,}{
                    $Q_{\textsc{rew}} := Q_{\textsc{rew}} \cup \{\langle q^{\prime}, 0 \rangle\}$\;
                }
            }
            \tcc{\textrm{rewriting step}}
            \ForEach{$A \subseteq \body{q}$}{
                \ForEach{$\sigma \in \Sigma$}{
                    \If{$\mathit{isApplicable}(\sigma,A,q)$}{
                        $q^{\prime} := \gamma_{A \cup \{\head{\sigma}\}}(q[A/\body{\sigma}])$\;
                        \If{$\mathit{notExists}(\langle q^{\prime},1 \rangle, Q_{\textsc{rew}})$}{
                            $Q_{\textsc{rew}} := Q_{\textsc{rew}} \cup \{\langle q^{\prime}, 1 \rangle\}$\;
                        }
                    }
                }
            }
        }
    }
    $Q_{\textsc{fin}} := \{q ~|~ \langle q,x \rangle \in Q_{\textsc{rew}} \textrm{~and~} x=1\}$\;
    \Return{$Q_{\textsc{fin}}$}
\end{algorithm}

\vspace{-2mm}

\textsc{Factorization Step.} The function
$\mathit{factorize}(q,\sigma)$, providing that there exists a subset
of $\body{q}$ which is factorizable w.r.t.~$\sigma$ (otherwise, the
query $q$ is returned), first selects such a set $S \subseteq
\body{q}$. Then, the query $q'$ is constructed by applying the MGU
$\gamma_{S}$ for $S$ on $q$.
Providing that there is no pair $\langle q^{\prime\prime},y
\rangle$, where $y \in \{0,1\}$, in $Q_{\textsc{rew}}$ such that
$q^{\prime}$ and $q^{\prime\prime}$ are the same (modulo bijective
variable renaming), the pair $\langle q^{\prime},0 \rangle$ is added
to $Q_{\textsc{rew}}$; the label $0$ keeps track of the queries
generated by the factorization step that must be excluded from the
final rewriting. This is carried out by the $\mathit{notExists}$
function.

\vspace{-2mm}

\textsc{Rewriting Step.} If there exists a pair $\langle q,
y\rangle$ and a TGD $\sigma \in \Sigma$ which is applicable to a set
of atoms $A \subseteq \body{q}$, then the algorithm constructs a new
query $q' = \gamma_{A \cup \{\head{\sigma}\}}(q[A /
\body{\sigma}])$, that is, the BCQ obtained from $q$ by replacing
$A$ with $\body{\sigma}$ and then applying the MGU for the set $A
\cup \{\head{\sigma}\}$. Providing that there is no pair $\langle
q^{\prime\prime},1 \rangle$ in $Q_{\textsc{rew}}$ such that
$q^{\prime}$ and $q^{\prime\prime}$ are the same (modulo bijective
variable renaming), the pair $\langle q^{\prime},1 \rangle$ is added
to $Q_{\textsc{rew}}$; the label $1$ keeps track of the queries
generated by the rewriting step which will be the final rewriting.

\begin{example}[Rewriting]
{\rm Consider the set $\Sigma$ of TGDs
\[
\begin{array}{rcl}
\sigma_1 &:& s(X) \, \rightarrow \exists Z\ \, t(X,X,Z)\\
\sigma_2 &:& t(X,Y,Z)\, \rightarrow\, r(Y,Z)
\end{array}
\]
and the query $q() \leftarrow t(A,B,C),r(B,C).$ {\sf TGD-rewrite}
first applies $\sigma_2$ to $\{r(B,C)\}$ since $\sigma_1$ is not
applicable. The query $q_1 : q() \leftarrow t(A,B,C), t(V^{1},B,C)$
is produced. Clearly, $\body{q_1}$ is factorizable w.r.t.~$\sigma_1$
and the query $q_2 : q() \leftarrow t(A,B,C)$ is obtained.
Now, $\sigma_1$ is applicable to $\{t(A,B,C)\}$ and the query $q_3 :
q() \leftarrow s(A)$ is obtained. The perfect rewriting constructed
by the algorithm is the set $\{q,q_1,q_3\}$.}
\label{exm:rewriting}
\end{example}

The next example shows that dropping the applicability condition,
then \textsf{TGD-rewrite} may produce unsound rewritings.

\begin{example}[Loss of soundness]
{\rm Suppose that we ignore the applicability condition during the
rewriting process. Consider the set $\Sigma$ of TGDs given in
Example~\ref{exm:rewriting}, and also the BCQ $q_1 : q() \leftarrow
t(A,B,c)$, where $c$ is a constant of $\Delta_c$.
A BCQ $q^{\prime}$ of the form  $q() \leftarrow s(V)$ is obtained,
where the information about the constant $c$ is lost. Consider now
the database $D = \{s(b),t(a,b,d)\}$ for $\R$. The query
$q^{\prime}$ maps to the atom $s(b)$ which implies that $D \models
q'$. However, the original query $q$ does not map to
$\chase{D}{\Sigma}$, and thus $D \cup \dep \not\models q$.
Therefore, any rewriting containing $q^{\prime}$ is not a sound
rewriting of $q$ given $\Sigma$.
Consider now the query $q'' : q() \leftarrow t(A,B,B)$. The same
query $q^{\prime}$ mapping to the atom $s(b)$ of $D$ is obtained.
However, during the construction of $\chase{D}{\dep}$ it is not
possible to get an atom of the form $t(X,Y,Y)$, where at positions
$t[2]$ and $t[3]$ the same value occurs. This implies that there is
no homomorphism that maps $q$ to $\chase{D}{\dep}$, and hence $D
\cup \dep \not\models q$. Therefore, any rewriting containing $q'$
is again unsound.} \label{exm:loss-soundness}
\end{example}

The applicability condition may prevent the generation of queries
that are vital to guarantee completeness of the rewritten query, as
shown by the following example. This is exactly the reason why the
factorization step is also needed.

\begin{example}[Loss of completeness]
{\rm Consider the set $\dep$ of TGDs
\[
\begin{array}{rcl}
\sigma_1 &:& p(X)\, \rightarrow\, \exists Y \, t(X,Y)\\
\sigma_2 &:& t(X,Y)\, \rightarrow\, s(Y)
\end{array}
\]
and the query $q : q() \leftarrow t(A,B),s(B).$ The only viable
strategy in this case is to apply $\sigma_2$ to $\{s(B)\}$, since
$\sigma_1$ is not applicable to $\{t(A,B)\}$ due to the shared
variable $B$. The query that we obtain is $q^{\prime} : q()
\leftarrow t(A,B), t(V^{1},B)$, where $V^{1}$ is a fresh variable.
Notice that in $q'$ the variable $B$ remains shared thus it is not
possible to apply $\sigma_1$. It is obvious that without the
factorization step there is no way to obtain the query
$q^{\prime\prime} : q() \leftarrow p(A)$ during the rewriting
process.
Now, consider the database $D = \{p(a)\}$. Clearly, $\chase{D}{\dep}
= \{p(a),t(a,z_1),s(z_1)\}$, and therefore $\chase{D}{\dep} \models
q$, or, equivalently, $D \cup \dep \models q$. However, the
rewritten query is not entailed by the given database $D$, since
$q''$ does not belong to it, which implies that it is not complete.}
\label{exm:loss-completeness}
\end{example}

We proceed now to establish soundness and completeness of the
proposed algorithm. Towards this aim we need two auxiliary technical
lemmas. The first one, which is needed for soundness, states that
once the chase entails the rewritten query constructed by the
rewriting algorithm, then the chase entails also the given query. In
the sequel, for brevity, given a BCQ $q$ over a schema $\R$ and a
set $\dep$ of TGDs over $\R$, we denote by $q_{\dep}$ the rewritten
query $\textsf{TGD-rewrite}(q,\dep)$.

\begin{lemma}\label{lem:sound-auxiliary-lemma}
Consider a BCQ $q$ over a schema $\R$, a database $D$ for $\R$, and
a set $\dep$ of TGDs over $\R$. If $\chase{D}{\dep} \models
q_{\dep}$, then $\chase{D}{\dep} \models q$.
\end{lemma}

\begin{proof}
The proof is by induction on the number of applications of the
rewriting step. We denote by $q_{\dep}^{[i]}$ the part of the
rewritten query $q_{\dep}$ obtained by applying $i$ times the
rewriting step.

\vspace{-2mm}

\textsc{Base Step.} Clearly, $q_{\dep}^{0} = q_{\dep}$, and the
claim holds trivially.

\vspace{-2mm}

\textsc{Inductive Step.} Suppose now that $\chase{D}{\dep} \models
q_{\dep}^{[i]}$, for $i \geq 0$. This implies that there exists $p
\in q_{\dep}^{[i]}$ such that $\chase{D}{\dep} \models p$, and thus
there exists a homomorphism $h$ such that $h(\body{p}) \subseteq
\chase{D}{\dep}$. If $p \in q_{\dep}^{[i-1]}$, then the claim
follows by induction hypothesis. The interesting case is when $p$
was obtained during the $i$-th application of the rewriting step
from a BCQ $p' \in q_{\dep}^{[i-1]}$, i.e., $q_{\dep}^{[i]} =
q_{\dep}^{[i-1]} \cup \{p\}$. By induction hypothesis, it suffices
to show that $\chase{D}{\dep} \models q_{\dep}^{[i-1]}$.

\vspace{-2mm}

Clearly, there exists a TGD $\sigma \in \dep$ of the form
$\phi(\insX,\insY) \ra \exists Z \, r(\insX,Z)$ which is applicable
to a set $A \subseteq \body{p'}$, and $p$ is such that $\body{p} =
\gamma(p'[A / \phi(\insX,\insY)])$, where $\gamma$ is the MGU for
the set $A \cup \{\head{\sigma}\}$. Observe that
$h(\gamma(\phi(\insX,\insY))) \subseteq \chase{D}{\dep}$, and hence
$\sigma$ is applicable to $\chase{D}{\dep}$; let $\mu = h \circ
\gamma$. Thus, $\mu'(r(\insX,Z)) \in \chase{D}{\dep}$, where $\mu'
\supset \mu$. We define the substitution $h' = h \cup \{\gamma(Z)
\ra \mu'(Z)\}$.

\vspace{-2mm}

Let us first show that $h'$ is a well-defined substitution. It
suffices to show that $\gamma(Z)$ is not a constant, and also that
$\gamma(Z)$ does not appear in the left-hand side of an assertion of
$h$. Towards a contradiction, suppose that $\gamma(Z)$ is either a
constant or appears in the left-hand side of an assertion of $h$. It
is easy to verify that in this case there exists an atom $\atom{a}
\in A$ such that at position $\pi_{\sigma}$ in $\atom{a}$ occurs
either a constant or a variable which is shared in $p'$. But this
contradicts the fact that $\sigma$ is applicable to $A$.
Consequently, $h'$ is well-defined.
It remains to show that the substitution $h' \circ \gamma$ maps
$\body{p'}$ to $\chase{D}{\dep}$, and thus $\chase{D}{\dep} \models
q_{\dep}^{[i-1]}$. Clearly, $\gamma(\body{p'} \setminus A) \subseteq
\body{p}$. Since $h(\body{p}) \subseteq \chase{D}{\dep}$, we get
that $h'(\gamma(\body{p'} \setminus A)) \subseteq \chase{D}{\dep}$.
Moreover,
\[
\begin{array}{rcl}
h'(\gamma(A)) &=& h'(\gamma(r(\insX,Z)))\\
&=& r(h'(\gamma(\insX)),h'(\gamma(Z)))\\
&=& r(\mu(\insX),\mu'(Z))\\
&=& \mu'(r(\insX,Z))\\
&\in& \chase{D}{\dep}.
\end{array}
\]
The proof is now complete. \qed
\end{proof}

The second auxiliary lemma, which is needed for completeness,
asserts that once the chase entails the rewritten query constructed
by the rewriting algorithm, then the given database also entails the
rewritten query.

\begin{lemma}\label{lem:complete-auxiliary-lemma}
Consider a BCQ $q$ over a schema $\R$, a database $D$ for $\R$, and
a set $\dep$ of TGDs over $\R$. If $\chase{D}{\dep} \models
q_{\dep}$, then $D \models q_{\dep}$.
\end{lemma}

\begin{proof}
We proceed by induction on the number of applications of the chase
step.

\vspace{-2mm}

\textsc{Base Step.} Clearly, $\apchase{0}{D}{\dep} = D$, and the
claim holds trivially.

\vspace{-2mm}

\textsc{Inductive Step.} Suppose now that $\apchase{i}{D}{\dep}
\models q_{\dep}$, for $i \geq 0$. This implies that there exists $p
\in q_{\dep}$ such that $\apchase{i}{D}{\dep} \models p$, and thus
there exists a homomorphism $h$ such that $h(\body{p}) \subseteq
\apchase{i}{D}{\dep}$. If $h(\body{p}) \subseteq
\apchase{i-1}{D}{\dep}$, then the claim follows by induction
hypothesis. The non-trivial case is when the atom $\atom{a}$,
obtained during the $i$-th application of the chase step due to a
TGD $\sigma \in \dep$ of the form $\phi(\insX,\insY) \ra \exists Z
\, r(\insX,Z)$, belongs to $h(\body{p})$. Clearly, there exists a
homomorphism $\mu$ such that $\mu(\phi(\insX,\insY)) \subseteq
\apchase{i-1}{D}{\dep}$ and $\atom{a} = \mu'(r(\insX,\insY))$, where
$\mu' \supseteq \mu$. By induction hypothesis, it suffices to show
that $\apchase{i-1}{D}{\dep} \models q_{\dep}$. Before we proceed
further, we need to establish an auxiliary technical claim.

\begin{claim}\label{cla:auxiliary-claim}
There exists a BCQ $p' \in q_{\dep}$ and a set of atoms $A \subseteq
\body{p'}$ such that $\sigma$ is applicable to $A$, and also there
exists a homomorphism $\lambda$ such that $\lambda(\body{p'}
\setminus A) \subseteq \apchase{i-1}{D}{\dep}$ and $\lambda(A) =
\atom{a}$.
\end{claim}

\begin{proof}
Clearly, there exists a set of atoms $B$ such that $h(\body{p}
\setminus B) \subseteq \apchase{i-1}{D}{\dep}$ and $h(B) =
\atom{a}$. Observe that the null value that occurs in $\atom{a}$ at
position $\pi_{\sigma}$ does not occur in $\apchase{i-1}{D}{\dep}$
or in $\atom{a}$ at some position other than $\pi_{\sigma}$.
Therefore, the variables that occur in the atoms of $B$ at
$\pi_{\sigma}$ do not appear at some other position. Consequently,
$B$ can be partitioned into the sets $B_1,\ldots,B_m$, where $m \geq
1$, and the following holds: for each $i \in [m]$, in the atoms of
$B_i$ at position $\pi_{\sigma}$ the same variable $V_i$ occurs, and
also $V_i$ does not occur in some other set $B \in
\{B_1,\ldots,B_m\} \setminus \{B_i\}$ or in $B_i$ at some position
other than $\pi_{\sigma}$. It is easy to verify that each set $B_i$
is factorizable w.r.t.~$\sigma$.

\vspace{-2mm}

Suppose that we factorize $B_1$. Then, the query $p_1 =
\gamma_1(p)$, where $\gamma_1$ is the MGU for $B_1$, is obtained.
Observe that $h$ is a unifier for $B_1$. By definition of the MGU,
there exists a substitution $\theta_1$ such that $h = \theta_1 \circ
\gamma_1$. Clearly,
\[
\begin{array}{rcl}
\theta_1(\body{p_1} \setminus \gamma_1(B)) &=&
\theta_1(\gamma_1(\body{p}) \setminus \gamma_1(B))\\
&=& h(\body{p} \setminus B)\\
&\subseteq& \apchase{i-1}{D}{\dep},
\end{array}
\]
and $\theta_1(\gamma_1(B)) = h(B) = \atom{a}$.

\vspace{-2mm}

Now, observe that the set $\gamma_1(B_2) \subseteq \body{p_1}$ is
factorizable w.r.t.~$\sigma$. By applying factorization we get the
query $p_2 = \gamma_2(p_1)$, where $\gamma_2$ is the MGU for
$\gamma_1(B_2)$. Since $\theta_1$ is a unifier for $\gamma_1(B_2)$,
there exists a substitution $\theta_2$ such that $\theta_1 =
\theta_2 \circ \gamma_2$. Clearly,
\[
\begin{array}{rcl}
\theta_2(\body{p_2} \setminus \gamma_2(\gamma_1(B))) &=&
\theta_2(\gamma_2(\body{p_1}) \setminus \gamma_2(\gamma_1(B)))\\
&=& \theta_1(\gamma_1(\body{p}) \setminus \gamma_1(B))\\
&=& h(\body{p} \setminus B)\\
&\subseteq& \apchase{i-1}{D}{\dep},
\end{array}
\]
and $\theta_2(\gamma_2(\gamma_1(B))) = \theta_1(\gamma_1(B)) = h(B)
= \atom{a}$.

\vspace{-2mm}

Eventually, by applying the factorization step as above, we will get
the BCQ
\[
p_m\ =\ \gamma_m \circ \ldots \circ \gamma_1(p),
\]
where $\gamma_j$ is the MGU for the set $\gamma_{j-1} \circ \ldots
\circ \gamma_1(B_j)$, for $j \in \{2,\ldots,m\}$ (recall that
$\gamma_1$ is the MGU for $B_1$), such that $\theta_m(\body{p_m}
\setminus \gamma_m \circ \ldots \circ \gamma_1(B)) \subseteq
\apchase{i-1}{D}{\dep}$ and $\theta_m(\gamma_m \circ \ldots \circ
\gamma_1(B)) = \atom{a}$.

\vspace{-2mm}

It is easy to verify that $\sigma$ is applicable to $A$. The claim
follows with $p' = p_m$, $A = \gamma_m \circ \ldots \circ
\gamma_1(B)$ and $\lambda = \theta_m$. \qed
\end{proof}

The above claim implies that during the rewriting process eventually
we will get a BCQ $p''$ such that $\body{p''} = \gamma(\body{p'}
\setminus A) \cup \gamma(\phi(\insX,\insY))$, where $\gamma$ is the
MGU for the set $A \cup \{\head{\sigma}\}$. It remains to show that
there exists a homomorphism that maps $\body{p''}$ to
$\apchase{i-1}{D}{\dep}$. Since $\lambda \cup \mu'$ is a
well-defined substitution, we get that $\lambda \cup \mu'$ is a
unifier for $A \cup \{\head{\sigma}\}$. By definition of the MGU,
there exists a substitution $\theta$ such that $\lambda \cup \mu' =
\theta \circ \gamma$. Observe that
\[
\begin{array}{rcl}
\theta(\body{p''}) &=& \theta(\gamma(\body{p'} \setminus A) \cup
\gamma(\phi(\insX,\insY)))\\
&=& (\lambda \cup \mu')(\body{p'} \setminus A) \cup (\lambda \cup
\mu')(\phi(\insX,\insY))\\
&=& \lambda(\body{p'} \setminus A) \cup \mu'(\phi(\insX,\insY))\\
&\subseteq& \apchase{i-1}{D}{\dep}.
\end{array}
\]
Consequently, the desired homomorphism is $\theta$, and the claim
follows. \qed
\end{proof}

We are now ready to establish soundness and completeness of the
algorithm \textsf{TGD-rewrite}.

\begin{theorem}\label{the:TGD-rewrite-sound-complete}
Consider a BCQ $q$ over a schema $\R$, a database $D$ for $\R$, and
a set $\dep$ of TGDs over $\R$. It holds that, $D \models q_{\dep}$
iff $D \cup \dep \models q$.
\end{theorem}

\begin{proof}
Suppose first that $D \models q_{\dep}$. Since $D \subseteq
\chase{D}{\dep}$, we get that $\chase{D}{\dep} \models q_{\dep}$,
and the claim follows by Lemma~\ref{lem:sound-auxiliary-lemma}.
Suppose now that $D \cup \dep \models q_{\dep}$. Since $q \in
q_{\dep}$, we get that $\chase{D}{\dep} \models q_{\dep}$, and the
claim follows by Lemma~\ref{lem:complete-auxiliary-lemma}. \qed
\end{proof}

Notice that the above result holds for arbitrary TGDs. However,
termination of \textsf{TGD-rewrite} is guaranteed if we consider
linear, sticky or sticky-join sets of TGDs since, during the
rewriting process, only finitely many queries (modulo bijective
variable renaming) are generated.

\begin{theorem}\label{the:TGD-rewrite-termination}
The algorithm \textsf{TGD-rewrite} terminates under linear, sticky
or sticky-join sets of TGDs.
\end{theorem}

Approaches such as those of~\cite{CDLLLR07} and~\cite{WCGP10} resort
to exhaustive factorizations of the atoms in the queries generated
by the rewriting algorithm. By factorizing a query $q$ we obtain a
subquery $q'$, that is, $q$ implies $q'$ (w.r.t. the given set of
TGDs). Observe that by computing the factorized query $q'$ we
eliminate unnecessary shared variables, in the body of $q$, due to
which the applicability condition is violated.
Consider for example the query $q^{\prime}$ of
Example~\ref{exm:loss-completeness}. By factorizing the body of $q'$
we obtain the query $q() \leftarrow t(A,B)$ which is a subquery
(w.r.t. to the given set $\dep$ of TGDs) of $q'$ (in this case
equivalent to $q'$), where the variable $B$ is no longer shared.
Thus, the rewriting step can now apply $\sigma_1$ to $\{t(A,B)\}$
and produce the query $q() \leftarrow p(A)$ which is needed to
ensure completeness.

\vspace{-2mm}

The exhaustive factorization produces a non-negligible number of
redundant queries as demonstrated by the simple example above. It is
thus necessary to apply a restricted form of factorization that
generates a possibly small number of BCQs that are necessary to
guarantee completeness of the rewritten query. This corresponds to
the identification of all the atoms in the query whose shared
existential variables come from the same atom in the chase, and they
can be thus unified with no loss of information. The key principle
behind our factorization process is that, in order to be applied,
there must exist a TGD that can be applied to the output of the
factorization.

\subsection{Exploiting Negative Constraints}\label{subsec:ncs}

It is well-known that negative constraints (NCs) of the form
$\forall \insX \, \phi(\insX) \ra \bot$ are vital for representing
ontologies. As already explained in
Subsection~\ref{subsec:additional-features}, given a database $D$
for a schema $\R$, a set $\dep$ of TGDs over $\R$, and a set
$\dep_\bot$ of NCs over $\R$, once the theory $D \cup \dep \cup
\dep_\bot$ is consistent, then we are allowed to ignore the NCs
since, for every BCQ $q$, $D \cup \dep \cup \dep_\bot \models q$ iff
$D \cup \dep \models q$. However, as shown in the following example,
by exploiting the given set of NCs it is possible to further reduce
the size of the final rewriting.

\begin{example}\label{exa:ncs}
Consider the TGD $\sigma : t(X),s(Y) \ra \exists Z \, p(Y,Z)$, the
NC $\nu : r(X,Y),s(Y) \ra \bot$, and the BCQ $q() \la
r(A,B),p(B,C)$. Clearly, due to the rewriting step, the query $p :
q() \la r(A,B),t(V^1),s(B)$ is obtained during the rewriting
process. However, this query is not really needed since, for any
database $D$ for $\R$, $D \not\models p$; otherwise, $D$ violates
the NC $\nu$ which is a contradiction since we always assume that
the theory $D \cup \{\sigma,\nu\}$ is consistent.
\end{example}

It is not difficult to show that, given a BCQ $q$, and a set $\dep$
of TGDs, if a query $p \in q_\dep$ is not entailed by
$\chase{D}{\dep}$, for an arbitrary database $D$, then any query $p'
\in q_\dep$ obtained during the rewriting process starting from $p$,
also it is not entailed by $\chase{D}{\dep}$. Assume now that the
set $\dep_\bot$ of NCs is part of the input. If we obtain a query $p
\in q_\dep$ such that there exists a homomorphism that maps
$\body{\nu}$, for some NC $\nu \in \dep_\bot$, to $\body{p}$, then
we can safely ignore $p$ since $\chase{D}{\dep}$ does not entail
$p$.

\vspace{-2mm}

From the above informal discussion, we conclude that we can further
reduce the size of the final rewriting by modifying our algorithm as
follows. During the execution of the rewriting algorithm
\textsf{TGD-rewrite} (see Algorithm~\ref{alg:tgd-rewrite}), after
the factorization step (resp., rewriting step) we check whether
there exists a homomorphism that maps $\body{\nu}$, for some NC
$\nu$ of the given set of NCs, to the body of the generated query
$q'$. If there exists such a homomorphism, then the pair
$\tup{q',0}$ (resp., $\tup{q',1}$) is not added to the set
$Q_{\textsc{rew}}$.
Furthermore, the pair $\tup{q,1}$ is added to $Q_{\textsc{rew}}$
(see the first line of the algorithm) only if there is no
homomorphism that maps $\body{\nu}$, for some NC $\nu$ of the given
set of NCs, to $\body{q}$. If there exists such a homomorphism, then
the algorithm terminates and returns the emptyset, which means that
$\chase{D}{\dep} \not\models q$, for every database $D$ for $\R$.

\section{Rewriting Optimization}
\label{sec:optimization}

It is common knowledge that the perfect rewriting obtained by
applying a backward-chaining rewriting algorithm (like {\sf
TGD-rewrite}) is, in general, not very well-suited for execution by
a DB engine due to the large number of queries to be evaluated. In
this section we propose a technique, called \emph{query
elimination}, aiming at optimizing the obtained rewritten query
under the class of linear TGDs. As we shall see, query elimination
(which is an additional step during the execution of the algorithm
{\sf TGD-rewrite}) reduces \emph{(i)} the number of BCQs of the
perfect rewriting, \emph{(ii)} the number of atoms in each query of
the rewriting as well as \emph{(iii)} the number of joins.
Note that in the rest of the paper we restrict our attention on
linear TGDs. Recall that linear TGDs are TGDs with just one atom in
their body. Since we also assume, as explained in the previous
section, TGDs with just one atom in their head, henceforth, when
using the term TGD, we shall refer to TGDs with just one body-atom
and one head-atom.

\vspace{-2mm}

By exploiting the given set of TGDs, it is possible to identify
atoms in the body of a certain query that are logically implied
(w.r.t. the given set of TGDs) by other atoms in the same query. In
particular, for each BCQ $q$ obtained by applying the rewriting step
of \textsf{TGD-rewrite}, the atoms of $\body{q}$ that are logically
implied (w.r.t. the given set of TGDs) by some other atom of
$\body{q}$ are eliminated.
Roughly speaking, the elimination of an atom from the body of a
query implies the avoidance of the construction of redundant queries
during the rewriting process. Thus, this step greatly reduces the
number of BCQs in the perfect rewriting.
Before going into the details, let us first introduce some necessary
technical notions.

\begin{definition}[Dependency Graph]
Consider a set $\dep$ of TGDs over a schema $\R$. The
\emph{dependency graph} of $\dep$ is a labeled directed multigraph
$\tup{N,E,\lambda}$, where $N$ is the node set, $E$ is the edge set,
and $\lambda$ is a labeling function $E \ra \dep$. The node set $N$
is the set of positions of $\R$. If there is a TGD $\sigma \in \dep$
such that the same variable appears at position $\pi_b$ in
$\body{\sigma}$ and at position $\pi_h$ in $\head{\sigma}$, then in
$E$ there is an edge $e = (\pi_b,\pi_h)$ with $\lambda(e) = \sigma$.
\label{def:updg}
\end{definition}

Intuitively speaking, the dependency graph of a set $\dep$ of TGDs
describes all the possible ways of propagating a term from a
position to some other position during the construction of the chase
under $\dep$. More precisely, the existence of a path $P$ from
$\pi_1$ to $\pi_2$ implies that it is possible (but not always) to
propagate a term from $\pi_1$ to $\pi_2$. The existence of $P$
guarantees the propagation of a term from $\pi_1$ to $\pi_2$ if, for
each pair of consecutive edges $e = (\pi,\pi')$ and $e' =
(\pi',\pi'')$ of $P$, where $e$ and $e'$ are labeled by the TGDs
$\sigma$ and $\sigma'$, respectively, the atom obtained during the
chase by applying $\sigma$ triggers $\sigma'$. To verify whether
this holds we need an additional piece of information, the so-called
\emph{equality type}, about the body-atom and the head-atom of each
TGD that occurs in $P$.

\begin{definition}[Equality Type]
Consider an atom $\atom{a}$ of the form $r(t_1,\ldots,t_n)$, where
$n \geq 1$. The \emph{equality type} of $\atom{a}$ is the set of
equalities
\begin{eqnarray*}
&\left\{r[i] = r[j]~|~t_i,t_j \not\in \Delta_c \textrm{~and~} t_i =
t_j\right\}&\\
&\bigcup&\\
&\left\{r[i] = c~|~c \in \Delta_c \textrm{~and~} t_i = c\right\}.&
\end{eqnarray*}
We denote the above set as $\mathit{eq}(\atom{a})$.
\label{def:equality-type}
\end{definition}

It is straightforward to see that, given a pair of TGDs $\sigma$ and
$\sigma'$, if $\mathit{eq}(\body{\sigma'}) \subseteq
\mathit{eq}(\head{\sigma})$, then there exists a substitution $\mu$
such that $\mu(\body{\sigma'}) = \head{\sigma}$. This allows us to
show that the atom obtained by applying $\sigma$ during the
construction of the chase triggers $\sigma'$.
Consequently, the existence of a path $P$ (as above) guarantees the
propagation of a term from $\pi_1$ to $\pi_2$ if, for each pair of
consecutive edges $e$ and $e'$ of $P$ which are labeled by $\sigma$
and $\sigma'$, respectively, $\mathit{eq}(\body{\sigma'}) \subseteq
\mathit{eq}(\head{\sigma})$.

\begin{example}[\it Dependency Graph]
{\rm Consider the set $\dep$ of TGDs
\[
\begin{array}{rcl}
\sigma_1 &:& p(X,Y) \rightarrow \exists Z r(X,Y,Z)\\
\sigma_2 &:& r(X,Y,c) \rightarrow s(X,Y,Y)\\
\sigma_3 &:& s(X,X,Y) \rightarrow p(X,Y).
\end{array}
\]
The equality type of the body-atoms and head-atoms of the TGDs of
$\dep$ are as follows:
\[
\begin{array}{rcl}
\mathit{eq}(\body{\sigma_1}) &=& \emptyset\\
\mathit{eq}(\head{\sigma_1}) &=& \emptyset\\
\mathit{eq}(\body{\sigma_2}) &=& \{r[3]=c\}\\
\mathit{eq}(\head{\sigma_2}) &=& \{s[2]=s[3]\}\\
\mathit{eq}(\body{\sigma_3}) &=& \{s[1]=s[2]\}\\
\mathit{eq}(\head{\sigma_3}) &=& \emptyset.
\end{array}
\]
The dependency graph of $\dep$ is shown in Figure~\ref{fig:updg}.}
\label{exm:updg}
\end{example}

We are now ready, by exploiting the dependency graph of a set of
TGDs, and the equality type of an atom, to introduce \emph{atom
coverage}.

\begin{figure}[t]
 \epsfclipon
  \centerline
  {\hbox{
  \epsfxsize=3.5cm
  \epsfysize=3.5cm
  \leavevmode
  \epsffile{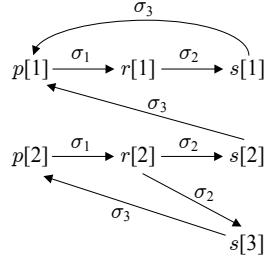}
  }} \epsfclipoff \caption{Dependency graph for Example~\ref{exm:updg}.}
  \label{fig:updg}
\end{figure}

\begin{definition}[Atom Coverage]
Consider a BCQ $q$ over a schema $\R$, and a set $\dep$ of TGDs over
$\R$. Let $\atom{a}$ and $\atom{b}$ be atoms of $\body{q}$, where
$\{t_1,\ldots,t_n\}$, for $n \geq 0$, is the set of shared variables
and constants that occur in $\atom{b}$. Also, let $G_\dep$ be the
dependency graph of $\dep$. We say that $\atom{a}$ \emph{covers}
$\atom{b}$ w.r.t.~$q$ and $\dep$, written as $\atom{a} \prec_{\dep}^{q}
\atom{b}$, if for each $i \in [n]$:
\begin{inparaenum}[(i)]
\item the term $t_i$ occurs also in $\atom{a}$, and 
\item if $t_i$ occurs in $\atom{a}$
and $\atom{b}$ at positions $\Pi_{\atom{a},i}$ and
$\Pi_{\atom{b},i}$, respectively, then, there exists an integer $k
\geq 2$ and a set of TGDs $\{\sigma_1,\ldots,\sigma_{k-1}\}
\subseteq \dep$, where $\mathit{eq}(\body{\sigma_1}) \subseteq
\mathit{eq}(\atom{a})$ and, for each $j \in [k-2]$,
$\mathit{eq}(\body{\sigma_{j+1}}) \subseteq
\mathit{eq}(\head{\sigma_j})$, such that,
for each $\pi \in \Pi_{\atom{b},i}$, in $G_\dep$ there exists a path
$\pi_{i_1} \pi_{i_2} \ldots \pi_{i_k}$, where $\pi_{i_1} \in
\Pi_{\atom{a},i}$, $\pi_{i_k} = \pi$, and
$\lambda((\pi_{i_j},\pi_{i_{j+1}})) = \sigma_j$, for each $j \in
[k-1]$.
\end{inparaenum}
\label{def:coverage}
\end{definition}

Condition \emph{(i)} ensures that by removing $\atom{b}$ from $q$ we
do not loose any constant, and also all the joins between $\atom{b}$
and the other atoms of $\body{q}$, except $\atom{a}$, are preserved.
Condition \emph{(ii)} guarantees that the atom $\atom{b}$ is
logically implied (w.r.t. $\dep$) by the atom $\atom{a}$, and
therefore can be eliminated.

\begin{lemma}\label{lem:atom-coverage}
Consider a BCQ $q$ over a schema $\R$, and a set $\dep$ of linear
TGDs over $\R$. Suppose that $\atom{a} \prec_{\dep}^{q} \atom{b}$,
where $\atom{a},\atom{b} \in \body{q}$, and $q'$ is the BCQ obtained
from $q$ by eliminating the atom $\atom{b}$. Then, $I \models q$ iff
$I \models q'$, for each instance $I$ that satisfies $\dep$.
\end{lemma}

\begin{proofsketch}
($\Rightarrow$) By hypothesis, there exists a homomorphism $h$ such
that $h(\body{q}) \subseteq I$. Since, by definition of $q'$,
$\body{q'} \subset \body{q}$, we immediately get that $h(\body{q'})
\subseteq I$, which implies that $I \models q'$.

\vspace{-2mm}

($\Leftarrow$) Conversely, there exists a homomorphism $h$ such that
$h(\body{q'}) \subseteq I$, and thus $h(\body{q} \setminus
\{\atom{b}\}) \subseteq I$. It suffices to show that there exists an
extension of $h$ which maps $\atom{b}$ to $I$. Since $\atom{a}
\prec_{\dep}^{q} \atom{b}$, it is not difficult to verify that there
exists an atom $\atom{c} \in I$ such that $\mathit{eq}(\atom{b}) =
\mathit{eq}(\atom{c})$, which implies that there exists a
substitution $\mu$ such that $\mu(\atom{b}) = \atom{c}$, and also
$\mu$ is compatible with $h$. Consequently, $(h \cup \mu)(\body{q})
\subseteq I$, and thus $I \models q$. \qed
\end{proofsketch}

An \emph{atom elimination strategy} for a BCQ is a permutation of
its body-atoms. Given a BCQ $q$ and a set $\dep$ of linear TGDs, the
set of atoms of $\body{q}$ that cover $\atom{a} \in \body{q}$
w.r.t.~$\dep$, denoted as $\mathit{cover}(\atom{a},q,\dep)$, is the
set $\{\atom{b}~|~\atom{b} \in \body{q} \textrm{~and~} \atom{b}
\prec_{\dep}^{q} \atom{a}\}$; when $q$ and $\dep$ are obvious from
the context, we shall denote the above set as
$\mathit{cover}(\atom{a})$. By exploiting the cover set of the atoms
of $\body{q}$, we associate to each atom elimination strategy $S$
for $q$ a subset of $\body{q}$, denoted
$\mathit{eliminate}(q,S,\dep)$, which is the set of atoms of
$\body{q}$ that can be safely eliminated (according to $S$) in order
to obtain a logically equivalent query (w.r.t.~$\dep$) with less
atoms in its body. Formally, $\mathit{eliminate}(q,S,\dep)$ is
computed by applying the following procedure; in the sequel, let $S
= [\atom{a}_1,\ldots,\atom{a}_n]$, where
$\{\atom{a}_1,\ldots,\atom{a}_n\} = \body{q}$:

\vspace{-2mm}

\begin{itemize}\itemsep-\parsep
\item[] $A := \emptyset$;
\item[] \textbf{foreach} $i := 1$ to $n$ \textbf{do}
\item[] \quad $\atom{a} := S[i]$;
\item[] \quad \textbf{if} $\mathit{cover}(\atom{a}) \neq \emptyset$ \textbf{then}
\item[] \quad \quad $A := A \cup \{\atom{a}\}$;
\item[] \quad \quad \textbf{foreach} $\atom{b} \in \body{q} \setminus
A$ \textbf{do}
\item[] \quad \quad \quad $\mathit{cover}(\atom{b}) := \mathit{cover}(\atom{b})
\setminus \{\atom{a}\}$;
\item[] return $A$.
\end{itemize}

\vspace{-2mm}

By exploiting the fact that the binary relation $\prec_{\dep}^{q}$
is transitive, it is possible to establish the uniqueness
(w.r.t.~the number of the eliminated atoms) of the atom elimination
strategy for a BCQ. In particular, the following lemma can be shown.

\begin{lemma}\label{lem:unique-elimination-strategy}
Consider a BCQ $q$ over a schema $\R$, and a set $\dep$ of linear
TGDs over $\R$. Let $S_1$ and $S_2$ be arbitrary elimination
strategies for $q$. It holds that, $|\mathit{eliminate}(q,S_1,\dep)|
= |\mathit{eliminate}(q,S_2,\dep)|$.
\end{lemma}

Since the elimination strategy for a query is unique (w.r.t.~the
number of the eliminated atoms), in the rest of this section we
refer to the set of atoms that can be safely eliminated from a query
$q$ (w.r.t.~a set $\dep$ of linear TGDs) by
$\mathit{eliminate}(q,\dep)$.

\vspace{-2mm}

We are now ready to describe how query elimination works. During the
execution of the rewriting algorithm \textsf{TGD-rewrite} (see
Algorithm~\ref{alg:tgd-rewrite}), after the factorization step and
the rewriting step the so-called \emph{elimination} step is applied.
In particular, the factorized query $q'$ obtained during the
factorization step is the query
$\mathit{eliminate}(\mathit{factorize}(q,\sigma),\dep)$, while the
rewritten query obtained during the rewriting step is the query
$\mathit{eliminate}(\gamma_{A \cup \{\head{\sigma}\}}(q[A /
\body{\sigma}]),\dep)$. Moreover, instead of adding the given query
$q$ in $Q_{\textsc{rew}}$, we add the eliminated query. In
particular, the first line of the algorithm is replaced by
$Q_{\textsc{rew}} := \tup{\mathit{eliminate}(q),1}$. An example of
query elimination follows.

\begin{example}[\it Query Elimination]
{\rm Consider the set $\dep$ of TGDs of Example~\ref{exm:updg}, and
the BCQ
\[
\begin{array}{rcl}
q() &\leftarrow& \underbrace{p(A,B)}_{\atom{a}},
\underbrace{r(A,B,C)}_{\atom{b}}, \underbrace{s(A,A,D)}_{\atom{c}}.
\end{array}
\]
Based on the Definition~\ref{def:coverage}, it is an easy task to
verify that $\mathit{cover}(\atom{a}) = \emptyset$,
$\mathit{cover}(\atom{b}) = \{\atom{a}\}$ and
$\mathit{cover}(\atom{c}) = \emptyset$. Therefore, the output of the
function $\mathit{eliminate}(q,\dep)$ is the singleton set
$\{\atom{b}\}$. Consequently, by applying the elimination step we
get the BCQ $q() \la p(A,B),s(A,A,D)$.}
\label{exm:query-reduction}
\end{example}

As already mentioned, the fact that an atom $\atom{a}$ covers some
atom $\atom{b}$, means that $\atom{b}$ is logically implied (w.r.t.
the given set of TGDs) by $\atom{a}$. However, as shown by the
following example, this fact is not also necessary for the
implication of $\atom{b}$ by $\atom{a}$.

\begin{example}[\it Atom Implication]
{\rm Consider the set $\dep$ of TGDs of Example~\ref{exm:updg}, and
the BCQ $q$
\[
\begin{array}{rcl}
q() &\leftarrow& \underbrace{r(A,A,c)}_{\atom{a}},
\underbrace{p(A,A)}_{\atom{b}},
\end{array}
\]
where $c$ is a constant of $\Delta_c$. Observe that $\atom{a}$ does
not cover $\atom{b}$ since, despite the existence of the paths $r[1]
s[1] p[1]$ and $r[2] s[3] p[2]$ in the dependency graph of $\dep$,
$\mathit{eq}(\body{\sigma_3}) \not\subseteq
\mathit{eq}(\head{\sigma_2})$.
However, $\atom{b}$ is logically implied (w.r.t. $\dep$) by
$\atom{a}$. In particular, for every instance $I$ that satisfies
$\dep$, if $I \models \atom{a}$, which implies that an atom of the
from $r(V,V,c)$ exists in $I$, then due to the TGDs $\sigma_2$ and
$\sigma_3$ there exists also an atom $p(V,V)$, and thus $I \models
\atom{b}$.
Note that such cases are identified by the C\&B
algorithm~\cite{DPT06}. Nevertheless, as already criticized in
Section~\ref{sec:related}, this requires to pay a price in the
number of queries in the rewritten query.}
\label{exm:atom-implication}
\end{example}

It is not difficult to see that the function \textit{eliminate} runs
in quadratic time in the number of atoms of $\body{q}$ (by
considering the given set of TGDs as fixed). In particular, to
compute the cover set of each body-atom of $q$ we need to consider
all the pairs of atoms of $\body{q}$. Note that the problem whether
a certain atom $\atom{a}$ covers some other atom $\atom{b}$ is
feasible in constant time since the given set of TGDs (and thus its
dependency graph) is fixed.

\vspace{-2mm}

The following result implies  that the rewriting algorithm
$\textsf{TGD-rewrite}^{\star}$, obtained from \textsf{TGD-rewrite}
by applying the additional step of elimination, is still sound and
complete.

\begin{theorem}\label{the:sound-complete}
Consider a BCQ $q$ over a schema $\R$, a database $D$ for $\R$, and
a set $\dep$ of linear TGDs over $\R$. Then, $D \models
\textsf{TGD-rewrite}^{\star}(\R,\dep,q)$ iff $D \cup \dep \models
q$.
\end{theorem}

\begin{proofsketch}
This result follows from the fact that the algorithm
\textsf{TGD-rewrite} is sound and complete under linear TGDs (see
Theorem~\ref{the:TGD-rewrite-sound-complete}) and
Lemma~\ref{lem:atom-coverage}. \qed
\end{proofsketch}

It is important to clarify that the above result does not hold if we
consider arbitrary TGDs (as in
Theorem~\ref{the:TGD-rewrite-sound-complete}). This is because the
proof of Lemma~\ref{lem:atom-coverage}, which states that atom
coverage implies logical implication (w.r.t. the given set of TGDs),
is based heavily on the linearity of TGDs.
Termination of $\textsf{TGD-rewrite}^{\star}$ follows immediately
from the fact that \textsf{TGD-rewrite} terminates under linear TGDs
(see Theorem~\ref{the:TGD-rewrite-termination}).

\section{Implementation and Experimental Setting}
\label{sec:experiments}

$\textsf{TGD-rewrite}$ (without the additional check described in
Subsection~\ref{subsec:ncs}) and the query elimination technique
presented in Section~\ref{sec:optimization} have been implemented in
the prototype system Nyaya~\cite{DOT*11} available at
\url{http://mais.dia.uniroma3.it/Nyaya}. The reasoning and query
answering engine is based on the IRIS Datalog
engine\footnote{\url{http://www.iris-reasoner.org/}.} extended to
support the FO-rewritable fragments of the {\plusminus} family. In
particular, we extended IRIS to natively support existential
variables in the head without introducing function symbols and to
support the constant $\mathit{false}$ as head of a rule (used to
represent negative constraints). Both IRIS and our extension are
implemented in Java.

\vspace{-2mm}

Since $\textsf{TGD-rewrite}$ is designed for reasoning over
ontologies with large ABoxes, we put ourselves in a similar
experimental setting such that of~\cite{PMH09}. Thus, we use
DL-Lite$_{\mathcal{R}}$ ontologies with a varying number of axioms.
The queries under consideration are based on canonical examples used
in the research projects where these ontologies have been developed.
VICODI (V) is an ontology of European history, and developed in the
EU-funded VICODI project\footnote{\url{http://www.vicodi.org}.}.
STOCKEXCHANGE (S) is an ontology for representing the domain of
financial institutions of the European Union. UNIVERSITY (U) is a
DL-Lite$_{\mathcal{R}}$ version of the LUBM
Benchmark\footnote{\url{http://swat.cse.lehigh.edu/projects/lubm/}.},
developed at Lehigh University, and describes the organizational
structure of universities. ADOLENA (A) (Abilities and Disabilities
OntoLogy for ENhancing Accessibility) is an ontology developed for
the South African National Accessibility Portal, and describes
abilities, disabilities and devices. The Path5 (P5) ontology is a
synthetic ontology encoding graph structures and used to generate an
exponential-blowup of the size of the rewritten queries. Recall that
the transformation of a set of TGDs into an equivalent set of
single-head TGDs with a single existential variable can introduce
auxiliary predicates and rules (see Lemmas~\ref{lem:one-head-atom}
and~\ref{lem:one-exist-variable}). The ontologies UX, AX and P5X are
equivalent ontologies to U, A and P5 where the auxiliary predicates
are considered part of the schema. These ontologies allow to study
the impact of such transformations on the size of the rewriting.

\vspace{-2mm}

We compared our implementation with two other rewriting-based query
answering systems for FO-rewritable ontologies:
QuOnto\footnote{\url{http://www.dis.uniroma1.it/quonto/}.}, based
on~\cite{CDLLLR07} and developed by the University of Rome La
Sapienza, and
Requiem\footnote{\url{http://www.comlab.ox.ac.uk/projects/requiem/home.html}.},
based on~\cite{PMH09} and developed by the Knowledge Representation
group of the University of Oxford.

\begin{table*}
\caption{Evaluation of Nyaya System.} \label{tab:eval} \centering
\tiny
\begin{tabular}{cc|c|c|c|c|c|c|c|c|c|c|c|c|}
\cline{3-14}
~ & ~ & \multicolumn{4}{c|}{Size} & \multicolumn{4}{c|}{Length} & \multicolumn{4}{c|}{Width} \\
\cline{3-14}
~&~&~&~&~&~&~&~&~&~&~&~&~&~\\
~ & ~ & QO & RQ & NY & NY$^{\star}$ & QO & RQ & NY & NY$^{\star}$ & QO & RQ & NY & NY$^{\star}$\\
~&~&~&~&~&~&~&~&~&~&~&~&~&~\\
\hline
\multicolumn{1}{|c|}{\multirow{5}{*}{V}}    & $q_1$ & 15     & 15    & 15    & 15  & 15     & 15     & 15      & 15    & 0      & 0      & 0      & 0    \\
\multicolumn{1}{|c|}{}                      & $q_2$ & 11     & 10    & 10    & 10  & 32     & 30     & 30      & 30    & 31     & 30     & 30     & 30   \\
\multicolumn{1}{|c|}{}                      & $q_3$ & 72     & 72    & 72    & 72  & 216    & 216    & 216     & 216   & 144    & 144    & 144    & 144  \\
\multicolumn{1}{|c|}{}                      & $q_4$ & 185    & 185   & 185   & 185 & 555    & 555    & 555     & 555   & 370    & 370    & 370    & 370  \\
\multicolumn{1}{|c|}{}                      & $q_5$ & 150    & 30    & 30    & 30  & 900    & 210    & 210     & 210   & 1,110  & 270    & 270    & 270  \\
\hline
\multicolumn{1}{|c|}{\multirow{5}{*}{S}}    & $q_1$ & 6      & 6     & 6     & 6   & 6       & 6      & 6      & 6     & 0      & 0      & 0      & 0   \\
\multicolumn{1}{|c|}{}                      & $q_2$ & 204    & 160   & 160   & 2   & 566     & 480    & 480    & 2     & 362    & 320    & 320    & 0   \\
\multicolumn{1}{|c|}{}                      & $q_3$ & 1,194  & 480   & 480   & 4   & 5,026   & 2,400  & 2,400  & 8     & 4,778  & 2,400  & 2,400  & 4   \\
\multicolumn{1}{|c|}{}                      & $q_4$ & 1,632  & 960   & 960   & 4   & 7,384   & 4,800  & 4,800  & 8     & 7,112  & 4,800  & 4,800  & 4   \\
\multicolumn{1}{|c|}{}                      & $q_5$ & 11,487 & 2,880 & 2,880 & 8   & 67,664  & 20,160 & 20,160 & 24    & 84,064 & 25,920 & 25,920 & 24  \\
\hline
\multicolumn{1}{|c|}{\multirow{5}{*}{U}}    & $q_1$ & 5      & 2     & 2     & 2   & 10      & 4      & 4      & 4  & 5       & 2      & 2       & 2    \\
\multicolumn{1}{|c|}{}                      & $q_2$ & 287    & 148   & 148   & 1   & 813     & 444    & 444    & 1  & 526     & 296    & 296     & 0    \\
\multicolumn{1}{|c|}{}                      & $q_3$ & 1,260  & 224   & 224   & 4   & 7,296   & 1,344  & 1,344  & 16 & 10,812  & 2,016  & 2,016   & 20   \\
\multicolumn{1}{|c|}{}                      & $q_4$ & 5,364  & 1,628 & 1,628 & 2   & 15,723  & 4,884  & 4,884  & 2  & 10,393  & 3,256  & 3,256   & 0    \\
\multicolumn{1}{|c|}{}                      & $q_5$ & 9,245  & 2,960 & 2,960 & 10  & 35,710  & 11,840 & 11,840 & 20 & 52,970  & 17,760  & 17,760 & 20   \\
\hline
\multicolumn{1}{|c|}{\multirow{5}{*}{A}}    & $q_1$ & 783    & 402 & 402 & 247 & 1,540   & 779   & 779 & 197   & 757     & 377  & 377 & 86   \\
\multicolumn{1}{|c|}{}                      & $q_2$ & 1,812  & 103 & 103 & 92  & 5,350   & 256   & 256 & 234   & 3,538   & 153 & 153  & 142   \\
\multicolumn{1}{|c|}{}                      & $q_3$ & 4,763  & 104 & 104 & 104 & 23,804  & 520   & 520 & 520   & 23,804  & 520 & 520  & 520   \\
\multicolumn{1}{|c|}{}                      & $q_4$ & 7,251  & 492 & 492 & 454 & 21,406  & 1,288 & 1,288 & 1,212   & 14,155  & 796  & 796 &     758  \\
\multicolumn{1}{|c|}{}                      & $q_5$ & 66,068 & 624 & 624 & 624 & 195,042 & 3,120  & 3,120  & 3,120 & 128,974  & 3,120  & 3,120 & 3,120  \\
\hline
\multicolumn{1}{|c|}{\multirow{5}{*}{P5}}   & $q_1$ & 14     & 6   & 6  & 6   & 14      & 6  & 6    & 6     & 0       & 0  & 0    & 0     \\
\multicolumn{1}{|c|}{}                      & $q_2$ & 86     & 10  & 10 & 10  & 156     & 16 & 16   & 16    & 70      & 6  & 6    & 6     \\
\multicolumn{1}{|c|}{}                      & $q_3$ & 538    & 13  & 13 & 13  & 1,413   & 29 & 29   & 29    & 900     & 16 & 16   & 16    \\
\multicolumn{1}{|c|}{}                      & $q_4$ & 3,620  & 15  & 15 & 15  & 14,430  & 44 & 44   & 44    & 10,260  & 29 & 29   & 29    \\
\multicolumn{1}{|c|}{}                      & $q_5$ & 25,256 & 16  & 16 & 16  & 107,484 & 60 & 60   & 60    & 103,361 & 44 & 44   & 44    \\
\hline
\multicolumn{1}{|c|}{\multirow{5}{*}{UX}}   & $q_1$ & 5      & 5     & 5     & 5    & 10      & 10       & 10       & 10 & 5      & 5      & 5      & 5   \\
\multicolumn{1}{|c|}{}                      & $q_2$ & 286    & 240   & 240   & 1    & 156     & 147      & 147      & 1  & 70     & 70     & 70     & 0   \\
\multicolumn{1}{|c|}{}                      & $q_3$ & 1,248  & 1,008 & 1,008 & 12   & 1,397   & 1,125    & 1,125    & 48 & 892    & 735    & 735    & 60  \\
\multicolumn{1}{|c|}{}                      & $q_4$ & 5,358  & 5,000 & 5,000 & 5    & 12,006  & 7,578    & 7,578    & 5  & 9,828  & 5,625  & 5,625  & 0   \\
\multicolumn{1}{|c|}{}                      & $q_5$ & 9,220  & 8,000 & 8,000 & 25   & 101,652 & 47,656   & 47,656   & 50 & 96,677 & 37,890 & 37,890 & 50  \\
\hline
\multicolumn{1}{|c|}{\multirow{5}{*}{AX}}   & $q_1$ & 783    & 782    & 782    & 555    & 1,543     & 1,541       & 1,541       & 1,084           & 763           & 761     & 761 & 529 \\
\multicolumn{1}{|c|}{}                      & $q_2$ & 1,812  & 1,781  & 1,781  & 1,737    & 3,589    & 3,528      & 3,528      & 3,514         & 3,576         & 3,516   & 3,516  & 3,401 \\
\multicolumn{1}{|c|}{}                      & $q_3$ & 4,763  & 4,752  & 4,752  & 4,741   & 27,705  & 23,760   & 23,760    & 23,760       & 23,824        & 23,815  & 23,815  & 23,694 \\
\multicolumn{1}{|c|}{}                      & $q_4$ & 7,251  & 7,100  & 7,100  & 6,564 & 7,739 & 7,578 & 7,578    & 6,178   & 5,744    & 5,625 & 5,625 & 5,201  \\
\multicolumn{1}{|c|}{}                      & $q_5$ & -      & -       & 76,032 & 76,032 & - & - & 81,173 & 81,173 & - & - & 95,942 & 95,942   \\
\hline
\multicolumn{1}{|c|}{\multirow{5}{*}{P5X}}  & $q_1$ & 14     & 14    & 14    & 14    & 14     & 14       & 14       & 14           & 0           & 0     & 0 & 0 \\
\multicolumn{1}{|c|}{}                      & $q_2$ & 86     & 77    & 77    & 66    & 156    & 147      & 147      & 121         & 70         & 70   & 70  & 55 \\
\multicolumn{1}{|c|}{}                      & $q_3$ & 530    & 390   & 390   & 329   & 1,397  & 1,125   & 1,125    & 925       & 892        & 735  & 735  & 596  \\
\multicolumn{1}{|c|}{}                      & $q_4$ & 3,476  & 1,953 & 1,953 & 1,644 & 12,006 & 7,578 & 7,578    & 6,263    & 9,828    & 5,625 & 5,625 & 4,619   \\
\multicolumn{1}{|c|}{}                      & $q_5$ & 23,744 & 9,766 & 9,766 & 8,219 & 101,652 & 47,656 & 47,656 & 39,531 & 96,677 & 37,890 & 37,890 & 31,312    \\
\hline
\hline
\end{tabular}
\end{table*}

Since $\textsf{TGD-rewrite}$, as well as the algorithms presented
in~\cite{CDLLLR07} and~\cite{PMH09}, are proven to be sound and
complete, the most relevant way of judging the quality of the
rewriting is the \emph{size} of the perfect rewriting, i.e., the
number of CQs in the perfect UCQ rewriting. In addition, we use two
additional metrics, namely, the \emph{length} of the rewriting,
i.e., the number of atoms in the perfect rewriting, and the
\emph{width}, i.e., the number of joins to be performed when the
rewritten query is executed. We believe these metrics to be more
appropriate than the number of symbols in the rewritten query used,
for example, in~\cite{PMH09}, since they allow to establish in a
more precise way the cost of executing the rewriting on a database
system. Table~\ref{tab:eval} reports the results of our
experiments\footnote{Additional data can be found on the Nyaya's Web
site.} while Table~\ref{tab:queries} shows the queries used in the
experiments. We use the symbol ``-'' to denote those cases where the
tool did not complete the rewriting within 15 minutes. By QO and RQ
we refer to the QuOnto and Requiem systems, respectively, while NY
and NY$^{\star}$ refer to Nyaya with factorisation and Nyaya with
both factorisation and query elimination, respectively. All the
tests have been performed on an Intel Core 2 Duo Processor at 2.50
GHz and 4GB of RAM. The OS is Ubuntu Linux 9.10 carrying a Sun JVM
Standard Edition with maximum heap size set at 2GB of RAM.

\begin{table*}
\caption{Test Queries}
\label{tab:queries}
\centering
\scriptsize
\begin{tabular}{|l|l|}
\hline
\multicolumn{1}{|c|}{TBox} & \multicolumn{1}{c|}{Queries} \\
\hline
\multicolumn{1}{|c|}{\multirow{5}{*}{V}} & $q_1(A) \leftarrow \mathit{ Location(A).}$\\
~ & $q_2(A,B) \leftarrow \mathit{Military\_Person(A), hasRole(B,A), related(A,C).}$\\
~ & $q_3(A,B) \leftarrow \mathit{Time\_Dependant\_Relation(A), hasRelationMember(A,B), Event(B).}$\\
~ & $q_4(A,B) \leftarrow \mathit{Object(A), hasRole(A,B), Symbol(B).}$\\
~ & $q_5(A) \leftarrow \mathit{Individual(A), hasRole(A,B), Scientist(B), hasRole(A,C), Discoverer(C), hasRole(A,D), Inventor(D).}$\\
\hline
\multicolumn{1}{|c|}{\multirow{5}{*}{S}} & $q_1(A) \leftarrow \mathit{ StockExchangeMember(A).}$\\
~ & $q_2(A,B) \leftarrow \mathit{ Person(A), hasStock(A,B), Stock(B).}$\\
~ & $q_3(A,B,C) \leftarrow \mathit{ FinantialInstrument(A), belongsToCompany(A,B), Company(B), hasStock(B,C), Stock(C).}$\\
~ & $q_4(A,B,C) \leftarrow \mathit{ Person(A), hasStock(A,B), Stock(B), isListedIn(B,C), StockExchangeList(C).}$ \\
~ & $q_5(A,B,C,D) \leftarrow \mathit{ FinantialInstrument(A), belongsToCompany(A,B), Company(B), hasStock(B,C), Stock(C),}$\\
~ & $\mathit{isListedIn(B,D), StockExchangeList(D).}$ \\
\hline
\multicolumn{1}{|c|}{\multirow{5}{*}{U(X)}} & $q_1(A) \leftarrow \mathit{ worksFor(A,B), affiliatedOrganizationOf(B,C).}$ \\
~ & $q_2(A,B) \leftarrow \mathit{ Person(A), teacherOf(A,B), Course(B).}$ \\
~ & $q_3(A,B,C) \leftarrow \mathit{ Student(A), advisor(A,B), FacultyStaff(B), takesCourse(A,C), teacherOf(B,C), Course(C).}$ \\
~ & $q_4(A,B) \leftarrow \mathit{ Person(A), worksFor(A,B), Organization(B).}$ \\
~ & $q_5(A) \leftarrow \mathit{ Person(A), worksFor(A,B), University(B), hasAlumnus(B,A).}$ \\
\hline
\multicolumn{1}{|c|}{\multirow{5}{*}{A(X)}} & $q_1(A) \leftarrow \mathit{ Device(A), assistsWith(A,B).}$ \\
~ & $q_2(A) \leftarrow \mathit{ Device(A), assistsWith(A,B), UpperLimbMobility(B).}$ \\
~ & $q_3(A) \leftarrow \mathit{ Device(A), assistsWith(A,B), Hear(B), affects(C,B), Autism(C).}$ \\
~ & $q_4(A) \leftarrow \mathit{ Device(A), assistsWith(A,B), PhysicalAbility(B).}$ \\
~ & $q_5(A) \leftarrow \mathit{ Device(A), assistsWith(A,B), PhysicalAbility(B), affects(C,B), Quadriplegia(C).}$ \\
\hline
\multicolumn{1}{|c|}{\multirow{5}{*}{P5(X)}} & $q_1(A) \leftarrow \mathit{ edge(A,B).}$ \\
~ & $q_2(A) \leftarrow \mathit{ edge(A,B), edge(B,C).}$ \\
~ & $q_3(A) \leftarrow \mathit{ edge(A,B), edge(B,C), edge(C,D).}$ \\
~ & $q_4(A) \leftarrow \mathit{ edge(A,B), edge(B,C), edge(C,D), edge(D,E).}$ \\
~ & $q_4(A) \leftarrow \mathit{ edge(A,B), edge(B,C), edge(C,D), edge(D,E), edge(E,F).}$ \\
\hline
\hline
\end{tabular}
\end{table*}

\vspace{-2mm}

As it can be seen, query elimination provides a substantial
advantage in terms of the size of the perfect rewriting for the
real-world ontologies A, U and S. In particular, for the queries
denoted as Q2 in U and S, our procedure eliminates all the redundant
atoms in the input query, and drastically reduces the number of
queries in the final rewriting.
On the other side, query elimination is not particularly effective in the synthetic test case P5 and
P5X, since these cases have been intentionally created in order to generate perfect rewritings of exponential size.

\section{Future Work}
\label{sec:conclusion}

We plan to investigate rewriting and optimization techniques for
sticky-join sets of TGDs, and alternative forms of rewriting such as
positive-existential queries. We also plan to develop improved
techniques for rewriting an ontological query into a non-recursive
Datalog program, rather than into a union of conjunctive queries
(recall the discussion in Section~\ref{sec:related}). While the
current approaches yield exponentially large non-recursive Datalog
programs, it is possible to rewrite queries and TBoxes into
non-recursive Datalog programs whose size is simultaneously
polynomial in the query and the TBox. This will be dealt in a
forthcoming paper.

\smallskip

\section*{Acknowledgments}
G. Gottlob's work was funded by the
European Research Council under the European Community's Seventh
Framework Programme (FP7/2007-2013)/ERC grant no.\ 246858 --
\mbox{DIADEM}. Gottlob gratefully acknowledges a Royal Society
Wolfson Research Merit Award.
G. Orsi and G. Gottlob also acknowledge the Oxford Martin School - Institute for the Future of Computing.
A. Pieris' work was funded by the EPSRC project ``Schema Mappings
and Automated Services for Data Integration and Exchange''
(EP/E010865/1).
We thank Micha\"{e}l Thomazo for his useful and constructive comments on the conference version of this paper.

\end{document}